\documentclass[aps,prx,reprint,superscriptaddress,amsmath,amssymb,longbibliography,10pt]{revtex4-2}

\usepackage{amsthm}
\usepackage{microtype}
\usepackage{booktabs}
\usepackage{tabularx}
\usepackage{array}
\usepackage{enumitem}
\usepackage[colorlinks=true,linkcolor=black,urlcolor=blue,citecolor=black]{hyperref}
\usepackage{orcidlink}
\usepackage{xspace}

\newtheorem{theorem}{Theorem}
\newtheorem{lemma}[theorem]{Lemma}

\newtheorem*{theoremEprime}{Theorem E$'$}
\newtheorem*{hypU}{Hypothesis (U)}

\newcommand{\Z}{\mathbb{Z}}
\newcommand{\E}{\mathbb{E}}
\newcommand{\Var}{\mathrm{Var}}
\newcommand{\Cov}{\mathrm{Cov}}
\newcommand{\ind}{\mathbf{1}}
\newcommand{\topp}{\mathrm{top}}
\newcommand{\msb}{\mathrm{msb}}
\newcommand{\Sim}{\cite{Sim26}\xspace}

\newcolumntype{L}{>{\raggedright\arraybackslash}X}

\begin{document}

\title{\texorpdfstring{Rigorous Statements and Proofs of the Lemmas in Simon's Algorithm\\ for the Dihedral Coset Problem and Their Underlying Hypothesis}{Rigorous Statements and Proofs of the Lemmas in Simon's Algorithm for the Dihedral Coset Problem and Their Underlying Hypothesis}}

\author{Yuchen Guo~\orcidlink{0000-0002-4901-2737}} \email{guo-yc23@mails.tsinghua.edu.cn} \affiliation{State Key Laboratory of Low Dimensional Quantum Physics and Department of Physics, Tsinghua University, Beijing 100084, China}

\author{Shuo Yang~\orcidlink{0000-0001-9733-8566}} \email{shuoyang@tsinghua.edu.cn} \affiliation{State Key Laboratory of Low Dimensional Quantum Physics and Department of Physics, Tsinghua University, Beijing 100084, China} \affiliation{Frontier Science Center for Quantum Information, Beijing 100084, China} \affiliation{Hefei National Laboratory, Hefei 230088, China}

\date{\today}

\begin{abstract}
In a recent preprint, Simon proposed a polynomial-time quantum algorithm for the Dihedral Coset Problem and rested the analysis on four lemmas.
Three of them carry only proof sketches, and this paper gives each of those three a statement that admits a single reading together with a complete proof.
Lemma 1 follows from an exact second-moment computation for the subset-sum counts, and it holds with probability tending to one in place of the constant originally claimed.
The amplitude bound of Lemma 3 follows from an exact Parseval identity on the cube of measurement outcomes and holds at every threshold with no \emph{well-behavedness} hypothesis, so that predicate leaves the argument entirely.
For Lemma 4, we compute both balls-in-bins covariances exactly and find that the second carries a term a fixed ball count leaves out.
The assumption that the distinguished group contains no faulty samples can also be dropped.
The two branch amplitudes share a signed prefactor, so the counting estimates control their difference and not the ratio the lemma states.
We prove the additive form and show that the closing argument consumes nothing more than that.
A single hypothesis survives all of this.
It asks that the partition into the two sides be fixed independently of the measured string, and the rule the algorithm gives for choosing that partition does not supply it.
Establishing these four lemmas therefore does not by itself establish the correctness of the algorithm.
\end{abstract}

\maketitle

\section{Introduction}

The Dihedral Coset Problem asks for a hidden shift.
One is given a supply of two-term superpositions $|0,\xi\rangle+|1,\xi+d\rangle$ with $\xi$ uniform on $\Z_N$ and $d$ fixed, and the task is to recover $d$.
It is the hidden subgroup problem for the dihedral group, and the standard Fourier method that solves the abelian cases in Shor's algorithm does not carry over \cite{Sho97,CvD10}.
Ettinger and H\o yer showed that polynomially many samples already determine $d$ information-theoretically, so the whole difficulty lies in reading $d$ out in polynomial time \cite{EH00}.
The best known algorithms are the subexponential sieves of Kuperberg \cite{Kup05,Kup13} and Regev \cite{Reg04b}.
Bacon, Childs and van Dam identified the optimal measurement for the problem and reduced its implementation to a classical subset-sum question \cite{BCD05}.

Much of the interest in the problem comes from lattices.
Ajtai's worst-case to average-case reduction placed lattice problems at the foundation of a family of cryptosystems \cite{Ajt96}.
Regev showed that a polynomial-time solver would give a polynomial-time quantum algorithm for the unique shortest vector problem \cite{Reg04}, and Brakerski, Kirshanova, Stehl\'e and Wen extended the reduction to Learning with Errors at polynomial approximation factors \cite{Reg09,BKSW18}.
The complexity of the Dihedral Coset Problem therefore sits directly beneath a large part of lattice-based cryptography, including the key-encapsulation scheme now being standardised \cite{NIST24}.
A claim of a polynomial-time quantum algorithm in this area invites close reading, and the recent lattice algorithm of Chen received exactly that before a step in it was found not to work \cite{Che24}.

Simon has recently posted a preliminary draft giving a polynomial-time quantum algorithm for the problem \Sim.
It works on $Q$ samples at once and sorts the $Q$ bits into groups of $m$.
The groups whose measurement outcome is all zero then build an unsigned count that both branches of the final interference share.
The analysis rests on four lemmas.
Lemma~2 is exact and carries a complete two-line proof, while the other three are marked \emph{Proof (Sketch)}.
This paper supplies statements for those three that admit a single reading, together with complete proofs.

\emph{Lemma 1} asserts that enough all-zero groups are collected.
Section~\ref{sec:l1} computes the first and second moments of their number exactly by means of two Fourier transforms, one on the cube of measurement outcomes and one on the group $\Z_N$ of subset sums.
The correlation between two groups turns out to be exponentially small, and Chebyshev then gives probability tending to one in place of the constant \Sim\ asks for.

\emph{Lemma 3} bounds the individual amplitudes.
Section~\ref{sec:l3} obtains the bound from an exact Parseval identity, which needs no independence property of the signs and no \emph{well-behavedness} property either.
The bound holds at every threshold, and the probabilities are taken directly on the distribution of records.
These proofs ask for the $A/B$ partition to be fixed independently of the measured string, which we state as Hypothesis~(U) in Section~\ref{sec:l3}.

\emph{Lemma 4} compares the two branch amplitudes, and its proof works on the counts that feed them.
The skeleton is a pair of balls-in-bins estimates, and Section~\ref{sec:l4} computes the covariance of each one exactly.
The first comes out with the exponent and the failure probability \Sim\ displays, while the second works with a random number of surviving states that the algorithm imposes.
We also remove the assumption that the distinguished group $g_a$ contains no faulty samples.
Carrying the counts back up to the amplitudes raises a question of form.
The two amplitudes of a pair share one signed prefactor, so the counting estimates deliver a bound on their difference while the lemma asks for their ratio.
That turns out to cost nothing, because the closing argument only ever adds up differences.
Section~\ref{sub:l4assemble} therefore proves the weaker statement and reaches the same exponent from it.

Table~\ref{tab:status} summarizes where each is treated.
The rule \Sim\ gives for choosing $A$ does not supply Hypothesis~(U), so establishing these lemmas does not by itself establish the correctness of the algorithm.
Results proved here are lettered, so a numbered lemma is always one of \Sim.

\begin{table}[htbp]

\renewcommand{\arraystretch}{1.15}
\centering
\caption{The four lemmas of \Sim\ and where each is treated below.}
\label{tab:status}
\begin{tabularx}{\linewidth}{@{}lLl@{}}
\toprule
In \Sim & What it asserts & Here \\
\midrule
Lemma~1 & Enough all-zero groups are collected
  & \ref{thm:C} \\
Lemma~2 & The $B$ side contributes a common factor
  & \S\ref{sec:setup} \\
Lemma~3 & Individual amplitudes are bounded & \ref{thm:E},
  \hyperlink{thm:Ep}{E$'$} \\
Lemma~4 & The two branch amplitudes are nearly equal & \ref{thm:G}, \ref{lem:H} \\
\bottomrule
\end{tabularx}
\end{table}

\section{Setup}\label{sec:setup}

\subsection{Problem setup}\label{sub:problem}

Write $N = 2^n$ and $\omega = e^{2\pi i/N}$, and let the unknown value be $d = d_1d_2\cdots d_n$ with least significant bit $d_n$.
A sample is the state
\begin{equation}
  \tfrac1{\sqrt2}\bigl(|0,\xi\rangle + |1,\xi+d \bmod N\rangle\bigr),
\end{equation}
with $\xi$ uniform and unknown and $d$ fixed.
With probability $1/(c'\log n)$ for a constant $c'$, a sample is faulty and is an $(n+1)$-bit classical string instead.
Faulty samples play no role before Section~\ref{sec:l4}.

Applying the $\Z_N$ Fourier transform to the second register and measuring it returns a value $y$ that is uniform on $\Z_N$ and known to us, and it collapses the first register to $|0\rangle + \omega^{yd}|1\rangle$.
Taking $Q$ samples and writing $Y = (y_1,\dots,y_Q)$ for the measured weights, the tensor product of the $Q$ qubits expands into
\begin{align}
  |\Psi\rangle &\propto \sum_{\phi\in\{0,1\}^Q}\omega^{d\,z(\phi)}\,|\phi\rangle , \\
  z(\phi) &= \sum_{i=1}^{Q}\phi_iy_i \bmod N .
\end{align}
Every basis state carries the same modulus, so all of the information sits in the phase, and that phase depends on $\phi$ only through the weighted sum $z(\phi)$.

Computing $z(\phi)$ into an ancilla and measuring its low $n-1$ bits returns a value $z'$ and leaves the top bit $h = \msb(z)$ in superposition, where $\msb$ denotes the most significant bit.
Since $\omega^{d\cdot2^{n-1}} = (-1)^{d_n}$, the surviving state is
\begin{align}
  |\Psi\rangle &\propto \sum_{h\in\{0,1\}}(-1)^{h\,d_n} \Bigl(\sum_{\phi\in Z'_h}|\phi\rangle\Bigr)|h\rangle , \\
  Z'_h &= \bigl\{\phi : z(\phi) = z' + 2^{n-1}h \bigr\} .
\end{align}
Here $h(\phi)$ is a function of $\phi$ and not a free register, and we write $Z' = Z'_0\sqcup Z'_1$ for the set of all $\phi$ compatible with the measured $z'$.
The bit $d_n$ now sits in the relative phase of $|h\rangle$.
If the register $|h\rangle$ stood alone, one Hadamard and one measurement would return $d_n$ with certainty, and a recursion on the remaining bits of $d$ would finish the problem.
However, the two branches are attached to disjoint sets of $\phi$, so they occupy orthogonal subspaces and cannot interfere.
Regev's way out is to erase the $\phi$ register with a subset-sum oracle \cite{Reg04}.
Simon's way out is to change basis, and the seven steps below carry that out.

\subsection{The algorithm in seven steps}

Throughout this paper the \emph{ket} holds the registers still in superposition and the \emph{record} holds the classical numbers already measured.
The record is $(Y, z', \phi', S, W', h')$, and any of its entries may be substituted into a formula.
Table~\ref{tab:steps} tracks both.

\begin{table}[htbp]

\renewcommand{\arraystretch}{1.15}
\centering
\caption{What survives each step. Erasures are reversible and produce no record.}
\label{tab:steps}
\begin{tabularx}{\linewidth}{@{}llL@{}}
\toprule
After & In the ket & New classical data \\
\midrule
Step 1 & $\phi$ & $Y=(y_1,\dots,y_Q)$ \\
Step 2 & $\phi,\,h$ & $z'$, the low $n-1$ bits of $z$ \\
Step 3 & $\phi,\,h,\,s_1\cdots s_G$ & none \\
Step 4 & $h,\,s_1\cdots s_G$ & $\phi'$, and $\phi$ is gone \\
Step 5 & $h,\,s_1\cdots s_a$ & $S=(\sigma'_j)_{j\in B}$ \\
Step 6 & $h,\,h^*$ & $W'=(s_1,\dots,s_{a-1},l_{s^*})$ \\
Step 7 & $h^*$ & $h'=h\oplus h^*$, and $h$ is erased \\
Final & none & the Hadamard measurement of $h^*$ \\
\bottomrule
\end{tabularx}
\end{table}

Steps~1 and 2 are the constructions of Section~\ref{sub:problem}, which produce the record entries $Y$ and $z'$ and leave the ket in the state displayed there.

\emph{Step 3, the group sums.} Partition $\{1,\dots,Q\}$ into $G = Q/m$ groups $g_1,\dots,g_G$ of $m = c\log n$ bits each, and write
\begin{equation}
  r_j(\phi) = \sum_{i\in g_j}\phi_iy_i \bmod N ,
  \qquad s_j(\phi) = \topp\bigl(r_j(\phi)\bigr)
\end{equation}
for the sum within group $j$ and its top $L = \log n$ bits.
Each $s_j$ is computed reversibly into its own $L$-bit register, giving
\begin{equation}
  |\Psi_3\rangle \propto \sum_{\phi\in Z'}(-1)^{h(\phi)d_n}\,
     |\phi\rangle\,|h(\phi)\rangle\,|s_1(\phi)\cdots s_G(\phi)\rangle .
\end{equation}
These registers survive the disappearance of $\phi$ at the next step, which keeps only the top $L$ bits of each $r_j$.

\emph{Step 4, the change of basis.} Apply a Hadamard to all $Q$ bits and measure, obtaining a classical string $\phi' = \phi'_1\cdots\phi'_G$, written in blocks by group.
The $\phi$ register is gone, and the amplitude of a surviving label $(h,\vec s\,)$ becomes a \emph{signed count} of the $\phi$ compatible with that label, each one weighted by the sign $(-1)^{\phi\cdot\phi'}$,
\begin{align}
  |\Psi_4\rangle &\propto \sum_{h,\vec s}\mathcal E[h,\vec s\,]\;|h,\vec s\,\rangle , \\
  \mathcal E[h,\vec s\,] &= (-1)^{hd_n}\!\!\!\!
     \sum_{\substack{\phi\in Z'\,:\ h(\phi)=h\\ s_j(\phi)=s_j\ \forall j}}\!\!\!\!
     (-1)^{\phi\cdot\phi'} .
\end{align}

\emph{Step 5, the $A/B$ split.} A group whose measured block $\phi'_j$ is the all-zero string contributes no sign, since $(-1)^{\phi_j\cdot0} = 1$.
Let $A$ consist of the first $a = n/\log n$ such groups and let $B$ hold the rest.
Lemma~1 of \Sim asserts that there are at least $a$ all-zero groups, and Section~\ref{sec:l1} treats that claim.
The $A$ side holds $am = cn$ sample bits.
Write $\phi = (\phi_A,\phi_B)$ and $\phi' = (\phi'_A,\phi'_B)$ for the strings cut into their $A$ and $B$ parts, so that $\phi'_A = 0$ by construction and
\begin{equation}
  (-1)^{\phi\cdot\phi'} \;=\; (-1)^{\phi_B\cdot\phi'_B} .
\end{equation}
Measure $s_j$ for every $j\in B$, obtaining $S = (\sigma'_j)_{j\in B}$, and the ket retains $(h,s_1\cdots s_a)$.

\emph{Step 6, manufacturing $h^*$.} Here we introduce a new register that represents the top bit of the $A$-side partial sum.
Write
\begin{equation}
  z^*(\phi) = \sum_{i\in A}\phi_iy_i \bmod N ,
  \qquad h^* = \msb(z^*)
\end{equation}
for the $A$-side partial sum and its top bit.
This step turns $h^*$ into an explicit register in five moves.

\begin{enumerate}[leftmargin=1.6em,itemsep=0.2em,topsep=0.3em]
\item Compute $s^* = \sum_{j\in A}s_j \bmod 2^L$ into a fresh $L$-bit register.
\item Erase $s_a$, which is now redundant, being
$s^* - \sum_{j<a}s_j \bmod 2^L$.
\item Compute a check bit $l_{s^*}$ that equals $1$ exactly when positions $2$ through $\log L$ of $s^*$ are all $1$, and post-select on $l_{s^*} = 0$.
This prohibits the carry from the bottom bits and makes $\msb(s^*) = \msb(z^*) = h^*$.
The post-selection succeeds with probability about $1-2/\log n$.
\item Split $s^*$ into its top bit $h^*$ and its low $L-1$ bits $\bar s^*$,
Hadamard $\bar s^*$, measure, and post-select on the all-zero outcome.
This succeeds with probability about $2/n$ and keeps all $n/2$ values of $\bar s^*$ in the sum with coefficient $+1$.
The bit $h^*$ itself is never measured, since we read it at the end.
\item Measure $W = (s_1,\dots,s_{a-1},l_{s^*})$, obtaining $W'$.
\end{enumerate}

The ket now holds $(h,h^*)$.

\emph{Step 7, transferring the phase.} Write $\tilde z^*$ for the low $n-1$ bits of $z^*$ and $z_B = \sum_{i\in B}\phi_iy_i$.
From $z^*+z_B \equiv z'+2^{n-1}h$ and $-2^{n-1}\equiv2^{n-1}\pmod{2^n}$,
\begin{equation}\label{eq:l2}
  z_B \;\equiv\; (z'-\tilde z^*) + 2^{n-1}h', \qquad h' \;=\; h\oplus h^* .
\end{equation}
The two top bits reach the $B$-side constraint only through $h'$.
Compute $h'$ with one CNOT, measure it, and use $h'$ together with $h^*$ to erase $h$ reversibly.
Once $h'$ is a measured constant the $B$ side no longer sees $h^*$, and the two branches share one $B$-side factor.
This is Lemma~2 of \Sim, and it is the only one of the four with a complete published proof.
The phase survives the substitution, since
\begin{equation}
  (-1)^{hd_n} \;=\; (-1)^{h'd_n}\,(-1)^{h^*d_n}
\end{equation}
and the first factor is a global constant once $h'$ is known.
The surviving state is
\begin{equation}
  |\Psi_7\rangle \;\propto\; \sum_{h^*\in\{0,1\}}(-1)^{h^*d_n}\,E_{h^*}\,|h^*\rangle ,
\end{equation}
where $E_{h^*}$ is the signed count of the $\phi$ compatible with the whole record and with that value of $h^*$.

\emph{The final step.} Hadamard $h^*$ and measure.
For a state $E_0|0\rangle + (-1)^{d_n}E_1|1\rangle$ the outcome equals $d_n$ with probability
\begin{equation}\label{eq:p}
  p \;=\; \frac{(E_0+E_1)^2}{2(E_0^2+E_1^2)} \;=\; \frac1{1+q^2},
  \qquad q \;=\; \frac{|E_0-E_1|}{|E_0+E_1|} ,
\end{equation}
so the whole algorithm reduces to showing that $E_0$ and $E_1$ are close.
Once $p$ exceeds $1/2$ by a non-negligible amount, polynomially many repetitions and a majority vote return $d_n$, and a recursion on the remaining bits of $d$ finishes the problem.

\subsection{The coefficient reduces to convolution}
\begin{table*}
\renewcommand{\arraystretch}{1.15}
\centering
\caption{The set families. Within each side, every entry carries the constraints
of the entry above it together with the one shown.}
\label{tab:sets}
\begin{tabularx}{0.86\linewidth}{@{}l c L@{}}
\toprule
Set & Fixed after & Constraint added \\
\midrule
\addlinespace[-0.3ex]
\multicolumn{3}{@{}l}{\itshape The $A$ side, where no sign is carried}\\
\addlinespace[0.3ex]
$\mathcal A_{z^*}$ & Step 4 & $z^*(\phi_A) = z^*$ \\
$\mathcal A'_{z^*,W'}$ & Step 6 & $s_j(\phi_A)=W'_j$ for $j<a$, together with
   $l_{s^*(\phi_A)}=0$ \\
\addlinespace[0.9ex]
\multicolumn{3}{@{}l}{\itshape The $B$ side, which carries every sign}\\
\addlinespace[0.3ex]
$\mathcal B_{(z',z^*),h}$ & Step 4 & $z_B(\phi_B) = z'+2^{n-1}h-z^*$ \\
$\mathcal B_{(z',z^*),S,h}$ & Step 5 & $s_j(\phi_B)=\sigma'_j$ for $j\in B$ \\
$\mathcal B'_{(z',\tilde z^*),S,h'}$ & Step 7 & $h$ replaced by the measured
   $h'$, so that by Eq.~\eqref{eq:l2} the index $z^*$ collapses to $\tilde z^*$ \\
\bottomrule
\end{tabularx}
\end{table*}
Two families of sets record how the constraints tighten as the measurements accumulate, one family on each side of the partition.
The two base members are fixed by the data known just after Step~4, namely the measured $z'$ together with a value of the index $z^*$ that ranges over $\Z_N$ as a summation variable and is never a measurement outcome,
\begin{align}
  \mathcal A_{z^*} &= \Bigl\{\phi_A : \sum_{i\in A}\phi_iy_i = z^*\Bigr\} , \\
  \mathcal B_{(z',z^*),h} &= \Bigl\{\phi_B : z^* + \sum_{i\in B}\phi_iy_i = z' + 2^{n-1}h\Bigr\} .
\end{align}
Each later measurement discards some of the strings that remain, and the notation records this by growing a subscript.
A prime marks the member that carries everything the record fixes.
On the $A$ side that member is $\mathcal A'_{z^*,W'}$, which collects the $\phi_A$ with $z^*(\phi_A) = z^*$ that also match the top blocks recorded in $W'$ at Step~6 and satisfy the post-selection $l_{s^*}=0$.
On the $B$ side it is $\mathcal B'_{(z',z^*),S,h'}$, which collects the $\phi_B$ that match the blocks $S$ measured at Step~5 and obey Eq.~\eqref{eq:l2} with the measured $h'$.
Table~\ref{tab:sets} collects all five members with the constraint each one adds.

The $A$-side member is indexed by the whole of $z^*$, while the $B$-side member is indexed by its low $n-1$ bits $\tilde z^*$ alone.
That collapse $\mathcal B'_{(z',z^*),S,h'}\rightarrow \mathcal B'_{(z',\tilde z^*),S,h'}$ is Lemma~2.
By Eq.~\eqref{eq:l2} the $B$ side is constrained by $(z'-\tilde z^*)$ and $h'$, so two values of $z^*$ differing only in the top bit $h^*$ give the same set of $\phi_B$, and therefore the same signed count.

The $B$-side signed count and the normalised per-index amplitude are
\begin{align}
  C_{(z',\tilde z^*),S,h'} &= \sum_{\phi_B\in\mathcal B'_{(z',\tilde z^*),S,h'}} (-1)^{\phi_B\cdot\phi'_B} , \\
  \alpha_{z^*} &= \lambda_4\,\bigl|\mathcal A'_{z^*,W'}\bigr|\, C_{(z',\tilde z^*),S,h'} ,
\end{align}
where $\lambda_4 = (E_0^2+E_1^2)^{-1/2}$ normalises the two branches at the end of Step~7.
The unnormalised amplitude of the branch $h^*$ is
\begin{equation}\label{eq:master}
  E_{h^*} \;=\; \sum_{\tilde z^*\in\{0,1\}^{n-1}}
     \bigl|\mathcal A'_{(h^*,\tilde z^*),W'}\bigr|\;
     C_{(z',\tilde z^*),S,h'} .
\end{equation}
The $A$-side factor is an unsigned count and every sign in the problem has been packed into $C$.
We write $\hat E_{h^*} = \lambda_4E_{h^*}$ for the normalised branch amplitude, so that $\hat E_0^2+\hat E_1^2 = 1$ and $\sum_{z^*\,:\,\msb(z^*)=h^*}\alpha_{z^*} = \hat E_{h^*}$.
Equation~\eqref{eq:p} may be read with either pair, since $q$ is a ratio.
The unnormalised $E_{h^*}$ is kept in what follows, since the weight of a record is proportional to $E_0^2+E_1^2$ and that quantity is identically $1$ in the hatted variables.
Subtracting the two branches gives the one line \Sim reduces to,
\begin{equation}\label{eq:diff}
  E_0 - E_1 \;=\; \sum_{\tilde z^*}
    \Bigl(\bigl|\mathcal A'_{(0,\tilde z^*),W'}\bigr|
        - \bigl|\mathcal A'_{(1,\tilde z^*),W'}\bigr|\Bigr)
    \cdot C_{(z',\tilde z^*),S,h'} .
\end{equation}
Each lemma guards one feature of Eq.~\eqref{eq:diff}.
Lemma~2 allows for the common factor $C$ outside the bracket.
Lemma~4 bounds the relative difference of the two counts, once they are aggregated over the low bits of $z^*$, by $O(n^{-((c-1)/2-1)})$.
Lemma~3 bounds each product $\alpha_{z^*}$ on its own, so that a small relative difference implies a small absolute one.

\subsection{Parameters}

Table~\ref{tab:params} lists the parameters that are fixed throughout our proof.
Our $\phi$ is written $b_1\cdots b_Q$ in \Sim, and our $\phi'$ is written $D$.
$n_j[r]$ counts the configurations of group $j$ whose subset sum equals $r$,
\begin{equation}
  n_j[r] \;=\; \#\bigl\{\phi_j\in\{0,1\}^m :
     \textstyle\sum_{i\in g_j}\phi_iy_i \equiv r\bigr\} .
\end{equation}

\begin{table}
\renewcommand{\arraystretch}{1.15}
\centering
\caption{Parameters. The sample count constant is written $K$, where \Sim\ writes
$k$, so that $k$ stays free for Fourier frequencies.}
\label{tab:params}
\begin{tabularx}{\linewidth}{@{}lLl@{}}
\toprule
Symbol & Meaning & Value \\
\midrule
$n$ & bit length & $N = 2^n$ \\
$c$ & group size constant & $\ge 12$ \\
$K$ & sample count constant & $> c$ \\
$Q$ & number of samples & $Kn^{c+1}$ \\
$m$ & bits per group & $c\log n$ \\
$G$ & number of groups & $Q/m$ \\
$a$ & number of $A$ groups & $n/\log n$ \\
$L$ & width of $s_j$ & $\log n$ \\
$\nu$ & length of one top block & $2^{\,n-L}$ \\
$Q_B$ & sample bits on the $B$ side & $(G-a)m$ \\
$1/(c'\log n)$ & faulty sample rate & $c'$ constant \\
\bottomrule
\end{tabularx}
\end{table}

\subsection{The probability space}\label{sub:prob}

One run produces the record $(Y,z',\phi',S,W',h')$, with $Y$ drawn at Step~1 and the rest supplied by the measurements of Steps~2 to 7.
Every probability below lives in the joint law of these.
The results about the samples, including Lemma~\ref{lem:F}, Theorem~\ref{thm:G} and Lemma~\ref{lem:H}, are proved with $Y$ uniform on $\Z_N^{\,Q}$.
Theorem~\ref{thm:E} and \hyperlink{thm:Ep}{Theorem~E$'$} concern the measured string and are proved with the rest of the record held fixed and $\phi'_B$ varying.
The bound they give is the same number for every value of the rest.
A bound holding under every conditioning holds under the joint law as well, so the failure probabilities of the two kinds live in one space and may be added.
The second kind deserves a warning, since $\phi'$ fixes the partition and the partition belongs to the record.
Holding the rest of the record fixed therefore already restricts which $\phi'_B$ remain available.
Letting $\phi'_B$ range over the whole cube is exactly Hypothesis~(U) of Section~\ref{sec:l3}, and Theorem~\ref{thm:E} and \hyperlink{thm:Ep}{Theorem~E$'$} are stated under it.

\section{Lemma 1}\label{sec:l1}

\begin{quote}
\textbf{Lemma 1} \Sim.
\emph{If $k > c$, then with constant probability, the measured value of $D = b'_1\ldots b'_Q$ will have at least $n/\log n$ groups that consist of all zeroes.}
\end{quote}

Here $k$ is our $K$ and $D$ is our $\phi'$.
This section proves the lemma in the stronger form that the probability tends to one.
The route is a direct moment computation on the number of all-zero groups, and the two transforms it needs are recorded first.

The argument in \Sim\ runs in two stages.
The first stage shows that zero blocks of $\phi'$ are at least as common as one blocks.
The second maps the records carrying fewer than $a$ all-zero groups into the records carrying at least $a$, by exchanging positions inside $\phi'$ and exchanging the low $n-2\log n$ bits of the corresponding weights $y_i$.
Those low bits are taken to be distinct across the sample, which makes the map invertible.
Comparing the two probabilities asks for more than invertibility.
Moving the low bits alone changes $\sum_i\phi_iy_i$, so the set $Z'_Y$ of strings compatible with the measured $z'$ is not carried onto its image counterpart.
The moment computation below reaches the conclusion without any such map.

\subsection{Two Fourier transforms}\label{sub:fourier}

Two different transforms appear below.
On the cube of measurement outcomes we use the Walsh--Hadamard transform \cite{ODo14},
\begin{equation}
  \widehat F(\phi') = \sum_{\phi\in\{0,1\}^Q}F[\phi]\,(-1)^{\phi\cdot\phi'} ,
\end{equation}
and on the group $\Z_N$ of subset sums we use
\begin{align}
  \hat f(k) = \sum_{r\in\Z_N}f[r]\,\omega^{-kr},
\end{align}
satisfying that
\begin{align}
  \sum_r f[r]\,\overline{g[r]} = \frac1N\sum_k \hat f(k)\,\overline{\hat g(k)} .
\end{align}
The second identity follows by expanding both sides and using $\frac1N\sum_k\omega^{k(r'-r)} = \delta_{rr'}$.

We also write $\langle f\rangle_k = \frac1N\sum_k f(k)$ and $\Var_k(f) = \langle f^2\rangle_k - \langle f\rangle_k^2$.
These measure the shape of a function on the frequency line and are not probabilistic quantities.

\subsection{One group}\label{sub:l1single}

Immediately after Step~4 the record consists of $(Y,z',\phi')$ and the basis states are labelled $(h,s_1\cdots s_G)$.
The amplitude at such a label is $\gamma\sum_{\phi\in\mathcal P}(-1)^{\phi\cdot\phi'}$ for a record-independent $\gamma$, where the consistency class is
\begin{align}
  \mathcal P(z, \vec s) = \bigl\{\phi\in\{0,1\}^Q : z(\phi)=z,\ s_i(\phi)=s_i\ \forall i\bigr\},
\end{align}
where $z = z'+2^{n-1}h$.
Every $\phi$ lies in exactly one class, so the classes partition the cube and $\sum_{\mathcal P}|\mathcal P| = 2^Q$.
Writing $\ind_{\mathcal P}(\phi)$ for the indicator function of a class
\begin{align}
    \ind_{\mathcal P}[\phi]=\begin{cases}1,& \phi \in \mathcal{P},\\ 0,&\text{otherwise},\end{cases}
\end{align}
we have
\begin{equation}
  \Pr[\phi'; \phi\in\mathcal{P}] = \gamma^2\Bigl(\sum_{\phi\in\mathcal P}(-1)^{\phi\cdot\phi'}\Bigr)^2
   = \gamma^2\bigl(\widehat{\ind_{\mathcal P}}(\phi')\bigr)^2 .
\end{equation}

\begin{lemma}\label{lem:A}
For every group $j$, every $Y$, and every fixed outcome of the earlier measurements, $\Pr[\phi'_j = 0 \mid \phi \in\mathcal{P}] \ge 2^{-m}$.
Summing over the classes, which removes the constraints imposed by $z'$, $h$ and $\vec s$,
\begin{equation}
  \Pr[\phi'_j = 0 \mid Y] = 2^{-2m}\sum_r n_j[r]^2 .
\end{equation}
\end{lemma}

\begin{proof}
Expanding the square gives
\begin{align}
    \bigl(\widehat{\ind_{\mathcal P}}(\phi')\bigr)^2 = \sum_{\phi,\psi\in\mathcal P}(-1)^{(\phi\oplus\psi)\cdot\phi'}.
\end{align}
Summing over all $\phi'$ leaves only the terms with $\phi\oplus\psi = 0$, so
\begin{align}
    \Pr[\phi\in\mathcal{P}] = \sum_{\phi'}\Pr[\phi', \phi\in\mathcal{P}] = \gamma^2 2^Q|\mathcal P|.   
\end{align}
Summing instead over the $2^{\,Q-m}$ strings with $\phi'_j = 0$ leaves the terms whose difference is supported inside $g_j$, which means that
\begin{align}
  \Pr[\phi_j'=0; \phi\in\mathcal{P}] = \sum_{\phi'\,:\,\phi'_j=0}\Pr[\phi'; \phi\in\mathcal{P}] = \gamma^2 2^{\,Q-m}A_{\mathcal P} ,
\end{align}
where
\begin{align}
  A_{\mathcal P} = \#\bigl\{(\phi,\psi)\in\mathcal P^2 : \phi\oplus\psi\ \text{supported in } g_j\bigr\} .
\end{align}
Dividing gives
\begin{align}
    \Pr[\phi'_j=0\mid \phi\in\mathcal{P}] = 2^{-m}A_{\mathcal P}/|\mathcal P|,
\end{align}
which is at least $2^{-m}$ because the diagonal $\phi=\psi$ already contributes $|\mathcal P|$.

Now let us sum over the classes.
Suppose $\phi$ and $\psi$ lie in the same class, so that $z(\phi)=z(\psi)$, and suppose they also agree outside $g_j$.
Then $z(\phi)-z(\psi) = r_j(\phi)-r_j(\psi)$, so $r_j(\phi) = r_j(\psi)$.
The constraint on $s_j$ is then satisfied automatically.
Choosing the bits outside $g_j$ freely,
\begin{equation}
  \sum_{\mathcal P}A_{\mathcal P}
   = \underbrace{2^{\,Q-m}}_{\text{bits outside } g_j}
     \times\underbrace{\sum_r n_j[r]^2}_{r_j(\phi_j)=r_j(\psi_j)} ,
\end{equation}
and dividing by $\sum_{\mathcal P}|\mathcal P| = 2^Q$ gives the stated form.
\end{proof}

\subsection{Two groups}\label{sub:l1pair}

\begin{lemma}\label{lem:B}
For $j\ne j'$ and every $Y$,
\begin{align}
  \Pr\bigl[\phi'_j = \phi'_{j'} = 0 \mid Y\bigr] = 2^{-4m}\,\frac1N\sum_k T_j(k)\,T_{j'}(k) ,
\end{align}
where
\begin{align}
  T_j(k) = \sum_s\bigl|\hat w_j^{(s)}(k)\bigr|^2
\end{align}
and $w_j^{(s)}[r] = n_j[r]\,\ind[\topp(r)=s]$ being the restriction of $n_j$ to one top block.
\end{lemma}

\begin{proof}
The character sum of Lemma~\ref{lem:A} applies with $\{\phi'_j=0\}$ replaced by $\{\phi'_j=\phi'_{j'}=0\}$.
There are $2^{\,Q-2m}$ such strings, therefore
\begin{align}
  \Pr\bigl[\phi'_j=\phi'_{j'}=0 \mid Y\bigr] = 2^{-2m}\cdot\frac{\sum_{\mathcal P}A^{(2)}_{\mathcal P}}{2^Q} ,
\end{align}
where
\begin{align}
  A^{(2)}_{\mathcal P} = \#\bigl\{(\phi,\psi)\in\mathcal P^2 : \phi\oplus\psi\ \text{supported in } g_j\cup g_{j'}\bigr\} .
\end{align}
Let $\phi$ and $\psi$ agree outside $g_j\cup g_{j'}$.
Every $s_i$ with $i\ne j,j'$ then agrees automatically, and
\begin{equation}
  z(\phi)-z(\psi) = \bigl[r_j(\phi)-r_j(\psi)\bigr]
    + \bigl[r_{j'}(\phi)-r_{j'}(\psi)\bigr] .
\end{equation}
So the pair lies in one class exactly when $r_j+r_{j'}$ is preserved (from $z$) and $\topp (r_j)$ and $\topp (r_{j'})$ are each preserved (from $s_j$ and $s_{j'}$).
In this case, the two groups may trade a common offset
\begin{equation}
  \delta := r_j(\phi)-r_j(\psi) = r_{j'}(\psi)-r_{j'}(\phi).
\end{equation}
Define the autocorrelation
\begin{align}
  R_j[\delta] = \!\!\sum_{\substack{r,r'\,:\ r-r'=\delta\\ \topp(r)=\topp(r')}}\!\! n_j[r]\,n_j[r'] = \sum_s\bigl(w_j^{(s)}\star w_j^{(s)}\bigr)[\delta],
\end{align}
where we define
\begin{align}
  (f\star f)[\delta] = \sum_r f[r]f[r-\delta] .
\end{align}
Pulling out the $2^{\,Q-2m}$ choices of the bits outside the two groups and sorting the rest by $\delta$, we obtain
\begin{equation}
\begin{split}
  \sum_{\mathcal P}A^{(2)}_{\mathcal P} &= 2^{\,Q-2m}\sum_{\delta\in\Z_N}R_j[\delta]\,R_{j'}[-\delta] \\
  &= 2^{\,Q-2m}\sum_{\delta\in\Z_N}R_j[\delta]\,R_{j'}[\delta] ,
\end{split}
\end{equation}
because $R$ is even.
Finally, we use the fact that
\begin{equation}
\begin{split}
  \widehat{f\star f}(k)
   &= \sum_\delta\sum_r f[r]f[r-\delta]\,\omega^{-k\delta} \\
   &\overset{r'=r-\delta}{=} \sum_{r,r'}f[r]f[r']\,\omega^{-k(r-r')} \\
   &= \bigl|\hat f(k)\bigr|^2 ,
\end{split}
\end{equation}
so $\widehat{R_j}(k) = T_j(k)$, and the $T_j(k)$ of Lemma~\ref{lem:B} is the Fourier transform of the autocorrelation $R_j[\delta]$.
The Parseval identity then turns $\sum_\delta R_j[\delta]R_{j'}[\delta]$ into $\frac1N\sum_k T_j(k)T_{j'}(k)$ and finishes the proof.
\end{proof}

Since $T_j$ is the Fourier transform of $R_j$, the inverse transform at $\delta=0$ gives $\langle T_j\rangle_k = \sum_r n_j[r]^2$.
Lemma~\ref{lem:A} therefore also reads $\Pr[\phi'_j=0\mid Y] = 2^{-2m}\langle T_j\rangle_k$, and the two lemmas have the same shape.

\subsection{Averaging over the weights}\label{sub:l1thm}
Every probability written so far has been conditional on one fixed $Y$, with the Hadamard measurement of Step~4 supplying the only randomness.
The weights are produced by the algorithm itself at Step~1, where the Fourier measurement on the second register of a sample returns a $y_i$ uniform on $\Z_N$.
The $Q$ samples are prepared independently, so
\begin{equation}\label{eq:Yunif}
  Y = (y_1,\dots,y_Q) \;\sim\; \mathrm{Unif}\bigl(\Z_N^{\,Q}\bigr) .
\end{equation}
Theorem~\ref{thm:C} below is a statement about a run of the whole algorithm, so its probability space carries $Y$ as well, and from here on $\E$ denotes that average.
Now we calculate the expected second moment of a single group's subset-sum count.
Write
\begin{equation}
  \beta \;=\; \frac{2^m-1}{N}
\end{equation}
for the chance that a fixed nonempty subset of one group hits a fixed value.
The result is
\begin{equation}\label{eq:uu}
\begin{split}
  \E\bigl[n_j[r]\,n_j[r']\bigr]
   &= \ind[r{=}r'{=}0]
      + \beta\bigl(\ind[r{=}0]+\ind[r'{=}0]\bigr) \\
   &\quad + \beta\,\delta_{rr'} + \frac{4^m-3\cdot2^m+2}{N^2} ,
\end{split}
\end{equation}
whose four terms come from the pairs where both subsets are empty, where exactly one is, where the two coincide and are nonempty, and where they differ and neither is empty (Appendix~\ref{app:B}).
Carrying Eq.~\eqref{eq:uu} through $T_j(k) = \sum_s\sum_{r,r'\in s}n_j[r]n_j[r']\,\omega^{-k(r-r')}$ gives
\begin{equation}\label{eq:ET}
\begin{split}
  \E T(k) &= 2^m + \beta_2 D_\nu(k)
     + \beta_1\,\mathrm{Re}\,\hat\ind_{B_0}(k), \\
  \beta_1 &= 2\beta, \qquad
  \beta_2 = \frac{(4^m-3\cdot2^m+2)\,2^L}{N^2} ,
\end{split}
\end{equation}
where $B_0 = \{0,1,\dots,\nu-1\}$ is the bottom top-block (i.e., those with zero top bits) and $D_\nu(k) = |\hat\ind_{B_0}(k)|^2$ is the Fej\'er kernel of order $\nu$.
Two averages of that kernel are also computed there
\begin{equation}\label{eq:fejer}
  \langle D_\nu\rangle_k = \nu,
  \qquad
  \langle D_\nu^2\rangle_k = \frac{\nu(2\nu^2+1)}{3} \le \nu^3 .
\end{equation}

\begin{theorem}\label{thm:C}
Let $\mathcal Z = \sum_{j=1}^G\ind[\phi'_j = 0]$ be the number of all-zero groups.
Under the hypothesis $K>c$ of \Sim, $\Pr[\mathcal Z\ge a]\to1$.
This is Lemma~1 of \Sim\ with probability tending to one in place of a constant.
\end{theorem}

\begin{proof}
We take five steps.

\emph{Step 1.} Lemmas~\ref{lem:A} and \ref{lem:B} hold for each fixed $Y$, so
\begin{align}
  \Pr[\phi'_j=0\mid Y] &= 2^{-2m}\langle T_j\rangle_k , \\
  \Pr[\phi'_j=\phi'_{j'}=0\mid Y] &= 2^{-4m}\langle T_jT_{j'}\rangle_k .
\end{align}
Different groups use disjoint weights and all weights are independent, so $\E[T_jT_{j'}](k) = \E T_j(k)\,\E T_{j'}(k)$ pointwise in $k$.
The groups are identically distributed, so writing $\E T$ for the common value
\begin{align}
  p_1 &:= \Pr[\phi'_j=0] = 2^{-2m}\langle\E T\rangle_k , \\
  p_2 &:= \Pr[\phi'_j=\phi'_{j'}=0] = 2^{-4m}\langle(\E T)^2\rangle_k .
\end{align}

\emph{Step 2.} In general $\langle(\E T)^2\rangle_k \ne \langle\E T\rangle_k^2$, and the difference measures the correlation between two groups
\begin{align}
  p_2-p_1^2 = 2^{-4m}\Var_k(\E T) , \\
  \eta := \frac{p_2-p_1^2}{p_1^2} = \frac{\Var_k(\E T)}{\langle\E T\rangle_k^2},
\end{align}
where $\eta$ is the normalised correlation coefficient of two groups being simultaneously all-zero.
In this sense, a question about probabilistic independence has become a question about the shape of one function.

\emph{Step 3.} Write $Z_j = \ind[\phi'_j=0]$ for the indicator of group $j$.
Each $Z_j$ is Bernoulli with parameter $p_1$ and every pair has covariance $p_1^2\eta$, so
\begin{equation}
\begin{split}
  \Var[\mathcal Z] &= \underbrace{\sum_j\Var[Z_j]}_{Gp_1(1-p_1)} + \underbrace{\sum_{j\ne j'}\Cov(Z_j,Z_{j'})}_{G(G-1)p_1^2\eta}  \\
  &\le\; Gp_1 + G^2p_1^2\eta .
\end{split}
\end{equation}
Since $\E[\mathcal Z] = Gp_1$, dividing by $\E[\mathcal Z]^2$ gives
\begin{equation}\label{eq:varZ}
  \frac{\Var[\mathcal Z]}{\E[\mathcal Z]^2}
   \;\le\; \frac1{\E[\mathcal Z]} + \eta .
\end{equation}
The first is the Poisson fluctuation of $G$ nearly independent coins, and the second is the excess contributed by their dependence.

\emph{Step 4.} It remains to bound $\eta$.
In Eq.~\eqref{eq:ET} the term $2^m$ is constant in $k$, and $\Var_k$ is invariant under adding a constant, so only the two bumps contribute.
From $(x+y)^2\le2x^2+2y^2$ and $\Var_k \le \langle\,\cdot^2\rangle_k$, we have
\begin{equation}
  \Var_k(\E T) \;\le\;
   2\beta_2^2\langle D_\nu^2\rangle_k
   + 2\beta_1^2\bigl\langle(\mathrm{Re}\,\hat\ind_{B_0})^2\bigr\rangle_k .
\end{equation}
By Eq.~\eqref{eq:fejer} the first average is at most $\nu^3$, and the second is at most $\langle|\hat\ind_{B_0}|^2\rangle_k = \nu$.
Using $\beta_2\le 4^m2^L/N^2$, $\beta_1 = 2\beta \le 2\cdot2^m/N$ and $\nu = N/2^L$, we bound the variance as
\begin{equation}
\begin{split}
  \beta_2^2\nu^3 &\le \frac{16^m}{N2^L}, \qquad
  \beta_1^2\nu \le \frac{4\cdot4^m}{N2^L}, \\
  \text{so}\quad \Var_k(\E T) &\le \frac{2\cdot16^m + 8\cdot4^m}{N2^L} .
\end{split}
\end{equation}
For the denominator, $\langle\E T\rangle_k\ge2^m$ gives $\langle\E T\rangle_k^2\ge4^m$.
This gives
\begin{equation}
\begin{split}
  \eta \; &\le\; \frac{2\cdot4^m+8}{N2^L} \;=\; \frac{2n^{2c}+8}{n\,2^n}  \\
  &\le\; \frac{2n^{2c-1}+8}{2^n} \;=\; 2^{-n+O(\log n)} ,
\end{split}
\end{equation}
using $2^m = n^c$, $2^L = n$ and $N = 2^n$.
The correlation between two groups is therefore exponentially small.

\emph{Step 5.} By Lemma~\ref{lem:A}, $\E[\mathcal Z]\ge G2^{-m} = (K/c)a$, so $a \le (c/K)\E[\mathcal Z]$ with $c/K<1$.
Chebyshev's inequality together with Eq.~\eqref{eq:varZ} gives
\begin{equation}
\begin{split}
  \Pr[\mathcal Z \le a]
   &\le\; \Pr\Bigl[\bigl|\mathcal Z-\E\mathcal Z\bigr| \ge \bigl(1-\tfrac cK\bigr)\E\mathcal Z\Bigr] \\
  &\le\; \Bigl(1-\frac cK\Bigr)^{-2} \Bigl(\frac{c}{Ka} + \eta\Bigr) ,
\end{split}
\end{equation}
where the last step also uses $1/\E[\mathcal Z]\le c/(Ka)$.
Here $c$ and $K$ are constants, $a = n/\log n\to\infty$, and $\eta\to0$ exponentially, so the right side tends to zero.
\end{proof}

Therefore, Lemma~1 holds with a conclusion stronger than \Sim\ claimed.

\section{Lemma 3}\label{sec:l3}

\begin{quote}
\textbf{Definition 1} \Sim.
\emph{For a given measured $z'$, and a resulting superposition $\psi$ of values of $h^*$ at the end of step 7, let $M$ be the set of values $(Y, D, W', S, h')$ measured (apart from $z'$) during the algorithm, let $T^+_M$ (resp., $T^-_M$) be the set of states $\phi$ consistent with $M$ and $h^*$ and with phase $1$ (resp., $-1$) resulting from measurement of $M$ (ignoring the constant phase $\omega^{yz'd}$).
Let $T_M = T^+_M\cup T^-_M$, and let $t_M$ (resp., $t^+_M$, $t^-_M$) be $|T_M|$ (resp., $|T^+_M|$, $|T^-_M|$).
Then $\psi$ is well-behaved for a value of $h^*$ if $|t^+_M - t^-_M| \ge \Omega(2^{-n/2}\sqrt{t_M})$.}

\medskip \textbf{Lemma 3} \Sim.
\emph{With probability $1-O(2^{-n})$ (over choices of $M$), the superposition $\psi$ at the end of step 7 is well-behaved for at least one value of $h^*$.
Moreover, if it's well-behaved for a particular $h^*$, then the probability (over choices of $M$) that for some $z^*$ consistent with that value of $h^*$ the amplitude $\alpha_{z^*} = \nu_4\sum_{\phi_A\in A_{z^*,W'}}C_{\bar z^*}$ has magnitude greater than $\Omega(2^{3n/2})$ is less than $O(2^{-n})$.}
\end{quote}

In our notation $D$ is $\phi'$, $\nu_4$ is $\lambda_4 = (E_0^2+E_1^2)^{-1/2}$, $\bar z^*$ is $\tilde z^*$, and $A_{z^*,W'}$ is $\mathcal A'_{z^*,W'}$, all of which appear in Table~\ref{tab:sets}.
The set $T_M$ of Definition~1 and its size $t_M$ have nothing to do with the Fourier-side quantity $T_j(k)$ of Lemma~\ref{lem:B}.

Both halves in \Sim\ rest on a variance computation for $t^\pm_M$, in which the signs $(-1)^{\phi\cdot\phi'}$ are taken to be pairwise independent as $\phi'$ varies.
However, a sign depends on $\phi$ only through $\phi_B$ since $\phi'_A = 0$, so two states of $T_M$ that share a $\phi_B$ carry the same sign.
Lemma~\ref{lem:D} below reads the same object as an exact Walsh transform, which asks nothing of the signs.

Separately from the signs, the proof of Lemma~3 in \Sim\ implies a hypothesis that the algorithm steps alone do not give.
It sums $\phi'_B$ over the whole cube $\{0,1\}^{Q_B}$ with the rest of the record held fixed.
The set $A$ consists of the \emph{first} $a$ all-zero groups, so the partition is itself a function of $\phi'$ and the admissible strings form a proper subset of the cube.
We take the freedom as a hypothesis and carry it in the statements.
Theorem~\ref{thm:E} and \hyperlink{thm:Ep}{Theorem~E$'$} are both stated under Hypothesis~(U), as is the result of Appendix~\ref{app:E}.

\begin{hypU}
The $A/B$ partition is fixed and does not vary with $\phi'$.
With the rest of the record held fixed, $\phi'_B$ therefore ranges over all of $\{0,1\}^{Q_B}$.
\end{hypU}

Here we summarize our proof of Lemma~3.
It is the only one of the four that involves both the size of an amplitude and the probability of a record.
Its second half is the one the argument consumes, and \Sim\ states that half conditionally for an $h^*$ at which the superposition is \emph{well-behaved}.
The first half supplies that predicate with high probability.
The proof below carries no such condition, because it runs a Parseval identity directly on the distribution of records.
The normalisation $\sqrt{E_0^2+E_1^2}$ that defines $\alpha_{z^*}$ is the weight of the record itself.
Appendix~\ref{app:E} records the first half regardless, in the sharper form that the failure probability is $\varepsilon^2$ at every threshold $\varepsilon$.

The second half is a statement about the probability of a record.
Once $\phi'$ is measured, the state left by Section~\ref{sec:setup} is
\begin{equation}
  E_0(\phi')\,|0\rangle \;+\; (-1)^{d_n}E_1(\phi')\,|1\rangle .
\end{equation}
The probability of that record is the squared norm of this unnormalised vector, so it is proportional to $E_0(\phi')^2 + E_1(\phi')^2$.
Theorem~\ref{thm:E} is proved directly on the record distribution, whose weights $E_0(\phi')^2+E_1(\phi')^2$ are not uniform in $\phi'$.
The constant of proportionality is fixed by $Q$, $N$, $L$ and the number of faulty samples, so it does not depend on the record and cancels from every ratio below.
That non-uniformity is not itself an obstacle, since the weight enters the numerator and the denominator of Step~4 alike and leaves the value unchanged.

\subsection{The signed count as a Walsh transform}\label{sub:l3D}

Recall that $\phi'_A = 0$ and $(-1)^{\phi\cdot\phi'} = (-1)^{\phi_B\cdot\phi'_B}$, so the whole content of $\phi'$ sits in its $B$-side part $\phi'_B$.
Define
\begin{equation}
  \mathcal N_{h^*}(\phi_B) \;=\;
   \#\bigl\{\phi_A : (\phi_A,\phi_B)\in T_M \text{ for that } h^*\bigr\}
\end{equation}
as the number of $A$-side completions of a given $B$-side configuration, which takes nonnegative integer values.

\begin{lemma}\label{lem:D}
$E_{h^*}(\phi') = \widehat{\mathcal N_{h^*}}(\phi'_B)$ (the Walsh--Hadamard transform on $\{0,1\}^{Q_B}$), and consequently
\begin{equation}\label{eq:parseval}
  \sum_{\phi'_B\in\{0,1\}^{Q_B}}E_{h^*}(\phi')^2 \;=\; 2^{Q_B}\sum_{\phi_B}\mathcal N_{h^*}(\phi_B)^2 .
\end{equation}
\end{lemma}

\begin{proof}
We first group the states of $T_M$ by their $B$-side part.
States in one group share a $\phi_B$ and therefore the same sign, which may be pulled outside the inner sum,
\begin{equation}\label{eq:walsh}
\begin{split}
  E_{h^*}(\phi') &= \sum_{\phi\in T_M}(-1)^{\phi_B\cdot\phi'_B} \\
  &= \sum_{\phi_B}\Bigl(\underbrace{\sum_{\phi_A} \ind\bigl[(\phi_A,\phi_B)\in T_M\bigr]}_{ =\;\mathcal N_{h^*}(\phi_B)}\Bigr)(-1)^{\phi_B\cdot\phi'_B} ,
\end{split}
\end{equation}
and the right side is the Walsh--Hadamard transform of $\mathcal N_{h^*}$ evaluated at $\phi'_B$.

Next, square and sum over $\phi'_B$,
\begin{equation}
\begin{split}
  \sum_{\phi'_B}E_{h^*}(\phi')^2
   &= \sum_{\phi_B,\psi_B}\mathcal N_{h^*}(\phi_B)\,\mathcal N_{h^*}(\psi_B) \\
   &\times \sum_{\phi'_B\in\{0,1\}^{Q_B}}(-1)^{(\phi_B\oplus\psi_B)\cdot\phi'_B}\\
   &=2^{Q_B}\sum_{\phi_B}\mathcal{N}_{h^*}(\phi_B)^2,
\end{split}
\end{equation}
where the inner sum is a character sum and only the diagonal term $\phi_B=\psi_B$ survives.
\end{proof}

\subsection{The amplitude bound}\label{sub:l3F}

Write $\Delta_{z^*}$ for the single term of Eq.~\eqref{eq:master} carried by one value of $z^*$
\begin{align}
  E_{h^*}(\phi') = \sum_{\tilde z^*\in\{0,1\}^{n-1}}\Delta_{z^*}(\phi') , \\
  \Delta_{z^*}(\phi') = \bigl|\mathcal A'_{z^*,W'}\bigr|\, C_{(z',\tilde z^*),S,h'} ,
\end{align}
and after normalisation the amplitude carried by $z^*$ is $\alpha_{z^*}(\phi') = \Delta_{z^*}(\phi')/\sqrt{E_0(\phi')^2+E_1(\phi')^2}$, which is the $\alpha_{z^*}$ of \Sim.

\begin{theorem}\label{thm:E}
Assume Hypothesis~(U), then for every threshold $\theta>0$,
\begin{equation}
\begin{split}
  \E\Bigl[\ \sum_{z^*}\alpha_{z^*}^2\ \Bigr] &\;\le\; 1 , \\
  \text{and hence}\quad
  \Pr\Bigl[\ \max_{z^*}\bigl|\alpha_{z^*}\bigr| > \theta\ \Bigr]
   &\;\le\; \theta^{-2} .
\end{split}
\end{equation}
Here $\E$ and $\Pr$ are defined over $\phi_B'\in {\{0, 1\}}^{Q_B}$.
Taking $\theta = 2^{3n/2}$ gives $2^{-3n}$, stronger than the $O(2^{-n})$ of Lemma~3 of \Sim.
\end{theorem}

\begin{proof}
We take five steps.

\emph{Step 1.} Once the total sum is determined, $z^* = z - \sum_{i\in B}\phi_iy_i$, where the right side depends only on $\phi_B$ beyond the fixed $z$.
So $z^*$ is a function of $\phi_B$, written as $z^*(\phi_B)$.

\emph{Step 2.} Let
\begin{equation}
\begin{split}
  \mathcal N_{z^*}(\phi_B) = \#\Bigl\{\phi_A :\ &(\phi_A,\phi_B)\in T_M \\
   &\text{and }\textstyle\sum_{i\in A}\phi_iy_i= z^*\Bigr\} ,
\end{split}
\end{equation}
which adds to $\mathcal N_{h^*}(\phi_B)$ the requirement that the $A$-side sum equal this particular $z^*$.
The question is whether the extra requirement removes anything.
If $z^*\ne z^*(\phi_B)$ then no $\phi_A$ qualifies and $\mathcal N_{z^*}(\phi_B) = 0$.
If instead $z^* = z^*(\phi_B)$, then every $\phi_A$ making $(\phi_A,\phi_B)$ lie in $T_M$ satisfies $\sum_{i\in A}\phi_iy_i = z - \sum_{i\in B}\phi_iy_i = z^*$ automatically.
Therefore
\begin{align}
    \mathcal N_{z^*}(\phi_B) = \mathcal N_{h^*}(\phi_B)\cdot \ind [z^*(\phi_B)=z^*],
\end{align}
which means that
\begin{align}
  \sum_{z^*}\mathcal N_{z^*}(\phi_B) &= \mathcal N_{h^*}(\phi_B) , \\
  \sum_{z^*}\mathcal N_{z^*}(\phi_B)^2 &= \mathcal N_{h^*}(\phi_B)^2.
\end{align}

\emph{Step 3.} Equation~\eqref{eq:walsh} applies to $\mathcal N_{z^*}$, giving $\Delta_{z^*}(\phi') = \widehat{\mathcal N_{z^*}}(\phi'_B)$, and then
\begin{equation}
\begin{split}
  \sum_{\phi'_B}\sum_{z^*}\Delta_{z^*}(\phi')^2
   &= 2^{Q_B}\sum_{z^*}\sum_{\phi_B}\mathcal N_{z^*}(\phi_B)^2 \\
   &\overset{\text{Step 2}}{=}
     2^{Q_B}\sum_{\phi_B}\mathcal N_{h^*=0}(\phi_B)^2 \\
   &\qquad{}+ 2^{Q_B}\sum_{\phi_B}\mathcal N_{h^*=1}(\phi_B)^2 \\
   &= \sum_{\phi'_B}\bigl(E_0(\phi')^2+E_1(\phi')^2\bigr) ,
\end{split}
\end{equation}
where the last equality is Lemma~\ref{lem:D} again.

\emph{Step 4.} The weight of a record is proportional to $E_0(\phi')^2+E_1(\phi')^2$ and the constant cancels.
The quantity $\alpha_{z^*}$ carries $\sqrt{E_0^2+E_1^2}$ in its denominator and is defined only where the weight is positive, so the average runs over those records alone,
\begin{widetext}
\begin{equation}
  \E\Bigl[\sum_{z^*}\alpha_{z^*}^2\Bigr]
   = \frac{\displaystyle\sum_{\phi'_B\,:\,E_0^2+E_1^2>0}
       \underbrace{\bigl(E_0(\phi')^2+E_1(\phi')^2\bigr)}_{\text{weight}}
       \cdot\underbrace{\frac{\sum_{z^*}\Delta_{z^*}(\phi')^2}
         {E_0(\phi')^2+E_1(\phi')^2}}_{=\ \sum_{z^*}\alpha_{z^*}(\phi')^2}}
      {\displaystyle\sum_{\phi'_B\in\{0,1\}^{Q_B}}
        \bigl(E_0(\phi')^2+E_1(\phi')^2\bigr)} .
\end{equation}
\end{widetext}
The two factors in each numerator term cancel, leaving $\sum_{\phi'_B\,:\,E_0^2+E_1^2>0}\sum_{z^*}\Delta_{z^*}(\phi')^2$.
The identity of Step~3 runs over every $\phi'_B$, and the two sums differ by the nonnegative contribution of the zero-weight records.
Dropping that contribution only decreases the value, so
\begin{equation}
  \E\Bigl[\sum_{z^*}\alpha_{z^*}^2\Bigr]
   \;\le\; \frac{\sum_{\phi'_B}\sum_{z^*}\Delta_{z^*}(\phi')^2}
      {\sum_{\phi'_B}\bigl(E_0(\phi')^2+E_1(\phi')^2\bigr)} \;=\; 1 .
\end{equation}

\emph{Step 5.} Write $V(\phi') = \sum_{z^*}\alpha_{z^*}(\phi')^2$, which is nonnegative with $\E[V]\le1$ by Step~4.
Pointwise $\max_{z^*}\alpha_{z^*}(\phi')^2 \le V(\phi')$, so the event $\{\max_{z^*}|\alpha_{z^*}|>\theta\}$ is contained in $\{V>\theta^2\}$.
Splitting $\E[V]$ according to whether $V$ exceeds $\theta^2$, we have
\begin{equation}
  \E[V] \;\ge\; \theta^2\Pr\bigl[V>\theta^2\bigr] ,
\end{equation}
so $\Pr[V>\theta^2] \le \E[V]/\theta^2 \le \theta^{-2}$.
\end{proof}

\subsection{The corollary}\label{sub:l3Fp}

The corollary to Lemma~3 asks the same question of the aggregate
\begin{equation}\label{eq:bridge}
  \alpha_{z^*_h} \;=\; \sum_{z^*\,:\ \topp(z^*)\,=\,z^*_h}\alpha_{z^*} ,
\end{equation}
a class containing $2^{\,n-L}$ values of $z^*$.
The proof of Theorem~\ref{thm:E} used only that $z^*$ is a function of $\phi_B$, so any coarser classification carries the same argument.

\hypertarget{thm:Ep}{}%
\begin{theoremEprime}
Assume Hypothesis~(U).
Let $\zeta$ be any function of $\phi_B$, with values in an arbitrary set.
Write $\Delta_\zeta(\phi')$ for the part of $E_{h^*}(\phi')$ contributed by the terms with $\zeta(\phi_B) = \zeta$, and $\alpha_\zeta = \Delta_\zeta/\sqrt{E_0^2+E_1^2}$.
Then for every $\theta>0$,
\begin{align}
  \E\Bigl[\ \sum_\zeta\alpha_\zeta^2\ \Bigr] \le 1 , \\
  \Pr\Bigl[\ \max_\zeta|\alpha_\zeta| > \theta\ \Bigr] \le \theta^{-2} .
\end{align}
\end{theoremEprime}

\begin{proof}
Repeat the five steps of Theorem~\ref{thm:E} with $z^*$ replaced by $\zeta$.
Step~1 is now the hypothesis on $\zeta$.
Step~2 uses only that hypothesis and gives $\mathcal N_\zeta(\phi_B) = \mathcal N_{h^*}(\phi_B)\ind[\zeta(\phi_B)=\zeta]$, which has at most one nonzero entry for fixed $\phi_B$.
Steps~3 to 5 never refer to the meaning of $\zeta$.
\end{proof}

Taking $\zeta(\phi_B) = z^*(\phi_B)$ recovers Theorem~\ref{thm:E}.
Taking $\zeta(\phi_B) = \topp(z^*(\phi_B))$ gives
\begin{align}
  \E\Bigl[\sum_{z^*_h}\alpha_{z^*_h}^2\Bigr] &\le 1 , \\
  \Pr\Bigl[\max_{z^*_h}\bigl|\alpha_{z^*_h}\bigr| > n^{3/2}\Bigr] &\le n^{-3} ,
\end{align}
well beyond the $O(1/n)$ the corollary asks for.

Theorem~\ref{thm:E} and \hyperlink{thm:Ep}{Theorem~E$'$} give the second half of Lemma~3 and its corollary under Hypothesis~(U), at every threshold and for either value of $h^*$.
Both of these go beyond the published statement.
The first half is in Appendix~\ref{app:E} and nothing above or below draws on it.

\section{Lemma 4}\label{sec:l4}

\begin{quote}
\textbf{Lemma 4} \Sim.
\emph{For any $z^*$, let $z^*_h$ be the most significant $\log n$ bits of $z^*$, and $z^*_l$ be the remaining bits.
Then for any two values $z^*_{h,0}$ and $z^*_{h,1}$ of $z^*_h$ that differ only in the most significant bit $h^*$, and for $c\ge12$, the amplitudes $\alpha_{z_{h,0}}$ and $\alpha_{z_{h,1}}$ (defined as in the corollary to Lemma~3) differ by an expected multiplicative factor $1\pm O(n^{-(((c-1)/2)-1)})$ with probability at least $1-1/n$.}
\end{quote}

The two amplitudes being compared are built out of the $A$-side counts.
Substituting $\alpha_{z^*} = \lambda_4|\mathcal A'_{z^*,W'}|C_{(z',\tilde z^*),S,h'}$ into the aggregate of Eq.~\eqref{eq:bridge},
\begin{equation}\label{eq:l4bridge}
  \alpha_{z^*_h} \;=\; \lambda_4\sum_{z^*_l\in\Z_\nu}
     \bigl|\mathcal A'_{(z^*_h,z^*_l),W'}\bigr|\;C_{(z',\tilde z^*),S,h'} .
\end{equation}
Two values of $z^*_h$ differing only in the top bit share every bit of $\tilde z^*$ that the $B$-side factor sees, so $C$ is the same function of $z^*_l$ for both, which is Lemma~2.
The comparison therefore sits entirely in the unsigned counts $|\mathcal A'|$, and the rest of this section is about them.

The skeleton in \Sim\ is a pair of balls-in-bins models, one at the level of $z^*_h$ and one at the level of $z^*_l$.
A mean and a variance are read off each model and Chebyshev is applied to those two numbers.
The balls are subset sums and are pairwise independent without being independent, which is all a variance needs, so both numbers the sketch reads off come out right.
Sections~\ref{sub:l4G} to \ref{sub:l4I} replace the two models by exact computations and recover the exponents and failure probabilities the paper displays.
The assumption that $g_a$ carries no faulty samples is removed along the way.

\medskip
The $A$ side factorises over groups, and $g_a$ sits differently in the record from the other $a-1$ groups
\begin{equation}
  \phi_A = \bigl(\underbrace{\phi_{g_1},\dots,\phi_{g_{a-1}}}_{
      \text{top blocks pinned by } W'},\ \underbrace{\phi_{g_a}}_{
      \text{top block free}}\bigr) .
\end{equation}
Three constraints define $\mathcal A'_{z^*,W'}$, namely $\topp(r_j) = W'_j$ for $j<a$, $l_{s^*}=0$, and $z^*(\phi_A) = z^*$.
Split the third one along $z^* = z^*_h\nu + z^*_l$.
For the high half, the top $L$ bits $s_a$ of the sum over the distinguished group $g_a$ determine $z^*_h$ by
\begin{equation}
  z^*_h = \sum_{j<a}W'_j + s_a(\phi_{g_a}) \bmod 2^L .
\end{equation}
Every $W'_j$ here is a measured constant, so fixing $z^*_h$ is the same as fixing $s_a$ and confines $\phi_{g_a}$ to the block $\topp(r_{g_a}) = s_a$.
For the low half, writing $r_j = s_j\nu+\rho_j$ for each group $j$ as in Section~\ref{sub:l4I}
\begin{equation}
  z^*_l \equiv \rho_{g_a}(\phi_{g_a}) + \sum_{j<a}\rho_j \pmod\nu .
\end{equation}
Enumerating $\phi_{g_a}$ first and counting the remaining groups afterwards turns the count into a double sum
\begin{equation}\label{eq:Asplit}
  \bigl|\mathcal A'_{(z^*_h,z^*_l),W'}\bigr|
   \;=\; \sum_{\phi_{g_a}\,:\ \topp(r_{g_a})=s_a}\!\!\!\! X_{z^*_l-\rho_{g_a}} ,
\end{equation}
where
\begin{equation}\label{eq:Xrho}
  X_\rho \;=\; \#\Bigl\{(\phi_{g_j})_{j<a} :\ \topp(r_j)=W'_j\ \ \forall j<a, \textstyle\sum_{j<a}\rho_j \equiv \rho \Bigr\}
\end{equation}
counts the configurations of the pinned groups whose low parts add up to $\rho$.
That bin count is taken up again in Section~\ref{sub:l4I}.

\subsection{An exact covariance}\label{sub:l4G}

We start from the covariance computation using Eq.~\eqref{eq:uu}.

\begin{lemma}\label{lem:F}
For any group $j$ and any $r,r'\in\Z_N$,
\begin{equation}
  \Cov\bigl(n_j[r],\,n_j[r']\bigr)
   = \frac{2^m-1}{N}\Bigl(\delta_{rr'} - \frac1N\Bigr) .
\end{equation}
Consequently the top-block count $X_\sigma = \sum_{r\,:\,\topp(r)=\sigma}n_j[r]$ has variance exactly $(2^m-1)2^{-L}(1-2^{-L})$.
\end{lemma}

\begin{proof}
The first moment behind Eq.~\eqref{eq:uu} is $\E[n_j[r]] = \ind[r{=}0] + \beta$, with $\beta = (2^m-1)/N$ as in Section~\ref{sub:l1thm}, so
\begin{equation}
\begin{split}
  \E\bigl[n_j[r]\bigr]\E\bigl[n_j[r']\bigr]
   &= \ind[r{=}0]\ind[r'{=}0] \\
   &\quad + \beta\bigl(\ind[r{=}0]+\ind[r'{=}0]\bigr) + \beta^2 .
\end{split}
\end{equation}
Subtracting this from Eq.~\eqref{eq:uu} leads to
\begin{equation}
  \beta\,\delta_{rr'}
   + \frac{(4^m-3\cdot2^m+2)-(2^m-1)^2}{N^2}
   = \beta\,\delta_{rr'} - \frac{2^m-1}{N^2}.
\end{equation}
Summing over $r$ and $r'$ inside one block $\topp(r)=\topp(r')=\sigma$ gives
\begin{equation}
\begin{split}
  \Var(X_\sigma) &= \frac{2^m-1}{N}\Bigl(\nu-\frac{\nu^2}{N}\Bigr)
   = (2^m-1)\,\frac\nu N\Bigl(1-\frac\nu N\Bigr) \\
   &= (2^m-1)\,2^{-L}\bigl(1-2^{-L}\bigr) . \qedhere
\end{split}
\end{equation}
\end{proof}

\subsection{Concentration of the block counts}\label{sub:l4H}

\begin{theorem}\label{thm:G}
Let $\kappa>0$ and let $g_j$ be any group.
With probability at least $1-2^L/\kappa^2$ over the samples of that group, every top block satisfies
\begin{equation}\label{eq:thmH}
  \bigl|X_\sigma - \E X_\sigma\bigr|
   \;\le\; \kappa\sqrt{(2^m-1)\,2^{-L}\bigl(1-2^{-L}\bigr)} .
\end{equation}
For $\sigma\ne0$ the mean is $\E X_\sigma = (2^m-1)2^{-L}$, and the bound takes the relative form
\begin{equation}\label{eq:thmHrel}
  \frac{\bigl|X_\sigma-\E X_\sigma\bigr|}{\E X_\sigma}
   \;\le\; \kappa\sqrt{\frac{2^L\bigl(1-2^{-L}\bigr)}{2^m-1}} .
\end{equation}
At the parameters $2^m = n^c$ and $2^L = n$ of \Sim, taking $\kappa = n$ gives a relative deviation at most $n^{-((c-1)/2-1)}$ with failure probability at most $1/n$.
Those are the exponent and the failure probability displayed in Lemma~4.
\end{theorem}

\begin{proof}
We take five steps.

\emph{Step 1.} The first moment behind Eq.~\eqref{eq:uu} is $\E[n_j[r]] = \ind[r{=}0]+\beta$, and a top block holds $\nu = N/2^L$ residues, so
\begin{align}
  \E X_\sigma = \sum_{r\,:\,\topp(r)=\sigma}\bigl(\ind[r{=}0]+\beta\bigr) = \ind[\sigma{=}0] + \nu\beta , \label{eq:blockmean} \\
  \nu\beta = (2^m-1)2^{-L} .
\end{align}
The term $\ind[\sigma{=}0]$ comes from the atom at $r=0$, since $\topp(0)=0$ places the empty subset in the block $\sigma = 0$.
The $2^m-1$ nonempty subsets divide evenly among the $2^L$ blocks, and the block $\sigma=0$ holds one more.

\emph{Step 2.} Lemma~\ref{lem:F} gives $\Var(X_\sigma) = (2^m-1)2^{-L}(1-2^{-L})$ for every $\sigma$.

\emph{Step 3.} Fix one $\sigma$.
The standard deviation $\mathrm{sd}(X_\sigma) = \sqrt{\Var(X_\sigma)}$ does not depend on $\sigma$, so one threshold serves all blocks, and Chebyshev gives
\begin{equation}
  \Pr\Bigl[\bigl|X_\sigma-\E X_\sigma\bigr|
     > \kappa\,\mathrm{sd}(X_\sigma)\Bigr]
   \;\le\; \frac{\Var(X_\sigma)}{\kappa^2\Var(X_\sigma)} \;=\; \frac1{\kappa^2} .
\end{equation}
Only finiteness of the variance is used, and nothing about the shape of the distribution of $X_\sigma$.

\emph{Step 4.} Let $B_\sigma$ be the event that block $\sigma$ leaves the range, for $\sigma = 0,1,\dots,2^L-1$.
We need to exclude the union.
Pointwise $\ind[\bigcup_\sigma B_\sigma] \le \sum_\sigma\ind[B_\sigma]$, since the left side taking the value $1$ forces at least one term on the right to do the same and the terms are nonnegative.
Taking expectations and using Step~3,
\begin{equation}
  \Pr\Bigl[\bigcup_\sigma B_\sigma\Bigr]
   \;\le\; \sum_{\sigma=0}^{2^L-1}\Pr[B_\sigma] \;\le\; \frac{2^L}{\kappa^2} ,
\end{equation}
and the complementary event is Eq.~\eqref{eq:thmH}.

\emph{Step 5.} For $\sigma\ne0$, Step~1 gives $\E X_\sigma = (2^m-1)2^{-L}$, and dividing the bound of Step~4 by it produces the relative form.
Substituting $2^m = n^c$ and $2^L = n$, the relative standard deviation is $\sqrt{n(1-1/n)/(n^c-1)} = n^{(1-c)/2}(1+O(n^{-1}))$.
Taking $\kappa = n$ turns this into $n^{\,1-(c-1)/2} = n^{-((c-1)/2-1)}$ with failure probability $2^L/\kappa^2 = n/n^2$.
\end{proof}

\subsection{Faulty samples}\label{sub:l4faulty}

Suppose $g_a$ carries $f_a$ faulty samples.
A faulty sample fixes its bit $\phi_i$ at a definite value and takes no part in the superposition, so the group has only $m-f_a$ free bits.
The count $n_{g_a}[r]$ then ranges over $2^{\,m-f_a}$ configurations, and the frozen bits contribute one fixed shift, which moves the whole distribution without changing its shape.
So Lemma~\ref{lem:F}, Theorem~\ref{thm:G} and Lemma~\ref{lem:H} below all hold with $2^m$ replaced by $2^{\,m-f_a}$.

Since $f_a\sim\mathrm{Bin}(m,p_f)$ with $p_f = 1/(c'\log n)$ and $m = c\log n$, the mean $\E f_a = mp_f = c/c'$ is a constant.
On average
\begin{equation}
\begin{split}
  \E\bigl[2^{\,m-f_a}\bigr] &= 2^m\,\E\bigl[2^{-f_a}\bigr] \\
  &= 2^m\Bigl(1-\frac{p_f}{2}\Bigr)^m \longrightarrow 2^m e^{-c/(2c')} ,
\end{split}
\end{equation}
indicating a loss of one constant factor.
Passing from a ratio of averages to an average of ratios costs one application of Cauchy--Schwarz and stays at the same order.

Theorem~\ref{thm:G} needs the high probability version.
From $\binom mk \le m^k/k!$,
\begin{equation}
  \Pr\bigl[f_a\ge k\bigr] \;\le\; \binom mk p_f^{\,k}
   \;\le\; \frac{(mp_f)^k}{k!} \;=\; \frac{(c/c')^k}{k!} ,
\end{equation}
and $k = \Theta(\log n/\log\log n)$ makes the right side smaller than $1/n$.
With probability $1-1/n$ we therefore have $2^{\,m-f_a} = n^{\,c}2^{-f_a} = n^{\,c-o(1)}$, so the relative deviation of Theorem~\ref{thm:G} picks up a factor $n^{o(1)}$ and the exponent in $n^{-((c-1)/2-1)}$ is unchanged.

So the assumption that $g_a$ carries no faulty samples can be dropped, and no constant probability has to be paid for it.
Both failure probabilities are of order $1/n$ and may be added, and the same applies to the other $A$ groups.

\subsection{The second model under the conditioning}\label{sub:l4I}
\begin{table}

\renewcommand{\arraystretch}{1.15}
\centering
\caption{The two models. They part company in the first row, where the second
one sees the record.}
\label{tab:twomodels}
\begin{tabularx}{\linewidth}{@{}lLL@{}}
\toprule
 & First, at $z^*_h$ & Second, at $z^*_l$ \\
\midrule
Balls & all $2^m$ strings $\phi_{g_a}$ of $g_a$
   & the $\phi_A$ \textbf{compatible with $W'$} \\
One bin & a top block of $r_{g_a}$ & a value of $z^*_l$ \\
Bins & $2^L = n$ & $\nu = 2^{\,n-L}$ \\
\bottomrule
\end{tabularx}
\end{table}

The two balls-in-bins arguments are not the same construction, as Table~\ref{tab:twomodels} shows.
The first model lets $\phi_{g_a}$ run over all $2^m$ values with the other groups absent, while the second fixes the value $\phi_{g_a}$ first and lets the remaining groups run.
The $A$ groups therefore fall into two classes in the second model.
The $a-1$ groups with $j<a$ are pinned by $W'$ to $\topp(r_j) = W'_j$ while the $m$ bits inside each of them stay free, and all of the balls come from there.
The group $g_a$ has its bits fixed and contributes one definite number.
The record $W' = (s_1,\dots,s_{a-1},l_{s^*})$ holds top blocks only for $j<a$, and $s_a$ is not among them.

Write
\begin{align}
  X_{s_j} &= \#\bigl\{\phi_{j}\in\{0,1\}^m : \topp(r_j) = s_j\bigr\} , \\
  P &= \prod_{j<a}X_{s_j} ,
\end{align}
so that $X_{s_j}$ is the block count of Theorem~\ref{thm:G} taken at the pinned block, and $P$ is the total number of $\phi_A$ compatible with $W'$.
Write $r_j = s_j\nu + \rho_j$ with $0\le\rho_j<\nu$.
Since $s_j\nu\equiv0\pmod\nu$, the low $n-L$ bits of $z^*$ come only from the low parts,
\begin{equation}
  z^*_l \;\equiv\; \underbrace{\rho_{g_a}}_{\text{a definite number}}
     + \sum_{j<a}\rho_j \pmod\nu .
\end{equation}
The term from $g_a$ has to be separated out.
Its effect is to shift every bin cyclically by one fixed amount and is independent of the other groups.

The bin count $X_\rho$ of Eq.~\eqref{eq:Xrho}, which runs over the pinned groups $j<a$ alone, is then a convolution of their low parts on $\Z_\nu$.
Restricting $n_j$ to the pinned block gives the window slice
\begin{align}
  u_j(\rho) \;:=\; n_j\bigl[s_j\nu+\rho\bigr] , 
\end{align}
representing the same numbers as $n_j$, taken on one block and reindexed.
Fixing $(\rho_1,\dots,\rho_{a-1})$, the number of compatible $\phi_A$ is $\prod_j n_j[s_j\nu+\rho_j] = \prod_j u_j(\rho_j)$.
Summing over all choices with $\sum_j\rho_j\equiv\rho$ gives
\begin{equation}
\begin{split}
  X_\rho &= \sum_{\rho_1+\cdots+\rho_{a-1}\equiv\rho\ (\mathrm{mod}\ \nu)} \prod_j u_j(\rho_j)  \\
  &=\; \bigl(u_1*\cdots*u_{a-1}\bigr)(\rho) ,
\end{split}
\end{equation}
a cyclic convolution on $\Z_\nu$.

\begin{lemma}\label{lem:H}
Suppose at least one group has $s_j\ne0$, and write $\mu = \E[P]/\nu$.
Split each bin count into
\begin{equation}
  X_\rho \;=\; \frac P\nu + Y_\rho ,
\end{equation}
the first piece being the share of each bin when there are $P$ balls in all, and the second how far bin $\rho$ departs from its share.
Then $\E X_\rho = \mu$, $\sum_\rho Y_\rho = 0$, the variables $Y_\rho$ and $P$ are uncorrelated, and
\begin{equation}\label{eq:covY}
  \Cov\bigl(Y_\rho,Y_{\rho'}\bigr) = \mu\,\delta_{\rho\rho'} - \frac\mu\nu .
\end{equation}
This is the covariance of $\nu\mu$ balls thrown uniformly into $\nu$ bins, and the proof of Lemma~4 assumes exactly that.
The conditioning on $W'$ enters only through the ball count $P$, and since $P = \prod_j X_{s_j}$ has independent factors,
\begin{equation}\label{eq:varP}
  \frac{\Var(P)}{\bigl(\E P\bigr)^2}
   = \prod_j\Bigl(1+\frac{\Var(X_{s_j})}{\bigl(\E X_{s_j}\bigr)^2}\Bigr) - 1
   \;\approx\; a\,\frac{2^L}{2^m} ,
\end{equation}
so the relative fluctuation of $P$ is about $\sqrt{a\,2^{\,L-m}} = n^{(2-c)/2}(\log n)^{-1/2}$.
\end{lemma}

A bin count fluctuates in two ways.
The ball count itself changes, which makes every bin rise and fall together.
The balls also scatter differently once the count is fixed, which makes the bins trade against one another.
Written separately, Eq.~\eqref{eq:covY} is the scattering and Eq.~\eqref{eq:varP} is the change in the count.

\begin{proof}
Bin counts are a convolution of the slices.
One Fourier transform turns that into a pointwise product, and different groups use independent weights, so the whole computation factors group by group.
Write $\hat f(k) = \sum_\rho f(\rho)e^{-2\pi ik\rho/\nu}$ for the transform on $\Z_\nu$, so that $\hat X(k) = \prod_j\hat u_j(k)$.
The zero frequency is the ball count, $\hat X(0) = \sum_\rho X_\rho = P$.
The proof consists of four steps.

\emph{Step 1, a slice with $s_j\ne0$.} Let $s_{j_0}\ne0$.
Of the four terms in Eq.~\eqref{eq:uu}, the atom and the two cross-atoms sit at $r=0$, and $\topp(0)=0$ places them in the block $s_j = 0$ alone.
Every $r$ in the block $s_{j_0}$ is nonzero, so only two terms survive,
\begin{align}
  \E\,u_{j_0}(\rho) &= \frac{2^m-1}{N} , \\
  \E\bigl[u_{j_0}(\rho)\,u_{j_0}(\rho')\bigr] &= \frac{2^m-1}{N}\,\delta_{\rho\rho'} + \frac{4^m-3\cdot2^m+2}{N^2} ,
\end{align}
a constant together with a diagonal term.
Transforming each leads to
\begin{align}
  \E\,\hat u_{j_0}(k) &= \beta\,\nu\,\delta_{k0},\label{eq:slice1}\\
  \E\bigl[\hat u_{j_0}(k)\overline{\hat u_{j_0}(k')}\bigr]
   &= \beta\,\nu\,\delta_{kk'}
     + \frac{4^m-3\cdot2^m+2}{N^2}\,\nu^2\,\delta_{k0}\delta_{k'0} .\label{eq:slice2}
\end{align}
Equation~\eqref{eq:slice1} vanishes at $k\ne0$, and Eq.~\eqref{eq:slice2} vanishes at $k\ne k'$, since then the two indices cannot both be $0$.

\emph{Step 2, multiplying over groups.} Distinct groups use disjoint weights and all weights are independent, so the expectation of a product factors into the product of the expectations,
\begin{align}
  \E\,\hat X(k) &= \prod_j\E\,\hat u_j(k) , \\
  \E\bigl[\hat X(k)\overline{\hat X(k')}\bigr] &= \prod_j\E\bigl[\hat u_j(k)\overline{\hat u_j(k')}\bigr] .
\end{align}
There are four cases in $(k,k')$.

\emph{(a) First moment.} $\E\,\hat u_{j_0}(k) = 0$ by Eq.~\eqref{eq:slice1}, so one factor of the product vanishes and $\E\,\hat X(k) = 0$.
Substituting into the inverse transform $X_\rho = \frac1\nu\sum_k\hat X(k)e^{2\pi ik\rho/\nu}$ and taking expectations, only $k=0$ survives as
\begin{equation}
  \E X_\rho = \frac1\nu\,\E\,\hat X(0) = \frac{\E P}{\nu} = \mu
\end{equation}
that is independent of $\rho$.

\emph{(b) Second moment.} By Eq.~\eqref{eq:slice2} one factor is again zero for $k\neq k'$, so $\E[\hat X(k)\overline{\hat X(k')}] = 0$.
As for $k = k'\ne0$, we compute $\E|\hat u_j(k)|^2$ group by group from
\begin{equation}
  \E\bigl|\hat u_j(k)\bigr|^2 = \sum_{\rho,\rho'\in\Z_\nu}
     \E\bigl[u_j(\rho)u_j(\rho')\bigr]\,e^{-2\pi ik(\rho-\rho')/\nu} ,
\end{equation}
taking the four terms of Eq.~\eqref{eq:uu} in turn.
At $k\ne0$ only the atom and the diagonal survive.
A group with $s_j = 0$ carries both and a group with $s_j\ne0$ carries only the diagonal, so in either case
\begin{equation}
  \E\bigl|\hat u_j(k)\bigr|^2 = (2^m-1)2^{-L} + \ind\bigl[s_j{=}0\bigr]
   = \E X_{s_j} \qquad(k\ne0) ,
\end{equation}
the right side being the block mean of Eq.~\eqref{eq:blockmean}.
Multiplying over $j$ and using independence to turn the product back into an expectation,
\begin{equation}
\begin{split}
  \E\bigl|\hat X(k)\bigr|^2 &= \prod_j\E X_{s_j} \\
  &= \E\Bigl[\prod_j X_{s_j}\Bigr] = \E P = \nu\mu \qquad(k\ne0) .
\end{split}
\end{equation}
If $k = k' = 0$, then $\hat X(0) = P$ and the second moment is $\E[P^2]$, used only in Step~4.

These four cases also show that $Y_\rho$ and $P$ are uncorrelated.
We have $P = \hat X(0)$ while $Y_\rho$ is a linear combination of the $\hat X(k)$ with $k\ne0$, and term by term
\begin{equation}
\begin{split}
  \Cov\bigl(\hat X(k),\hat X(0)\bigr) &= \E\bigl[\hat X(k)\overline{\hat X(0)}\bigr] - \E\,\hat X(k)\;\overline{\E\,\hat X(0)} \\
  &= 0 \qquad(k\ne0) .
\end{split}
\end{equation}
The variance of a single bin is therefore the sum of the two pieces with no cross term,
\begin{equation}
  \Var\bigl(X_\rho\bigr) = \frac{\Var(P)}{\nu^2} + \mu\Bigl(1-\frac1\nu\Bigr) ,
\end{equation}
the first piece from the ball count moving and the second from the balls scattering.

\emph{Step 3, back to the bins.} Splitting the inverse transform by frequency,
\begin{equation}
\begin{split}
  X_\rho &= \frac1\nu\sum_k\hat X(k)\,e^{2\pi ik\rho/\nu} \\
  &= \underbrace{\frac{\hat X(0)}\nu}_{=\;P/\nu} + \underbrace{\frac1\nu\sum_{k\ne0}\hat X(k)\,e^{2\pi ik\rho/\nu}}_{ =\;Y_\rho} .
\end{split}
\end{equation}
Since $\E Y_\rho = 0$, expanding both factors and using $\overline{\hat X(k')} = \hat X(-k')$ gives the covariance of two bins $\rho$ and $\rho'$ as
\begin{align}
  \Cov\bigl(Y_\rho,Y_{\rho'}\bigr)&=\E\bigl[Y_\rho Y_{\rho'}\bigr]\\
    &= \frac1{\nu^2}\sum_{k\ne0}\sum_{k'\ne0} \E\bigl[\hat X(k)\overline{\hat X(k')}\bigr]\,e^{2\pi i(k\rho-k'\rho')/\nu}\\
    &= \frac{\nu\mu}{\nu^2}\sum_{k\ne0}e^{2\pi ik(\rho-\rho')/\nu} \\
    &= \frac\mu\nu\bigl(\nu\,\delta_{\rho\rho'}-1\bigr) = \mu\,\delta_{\rho\rho'} - \frac\mu\nu ,
\end{align}
where only the terms with $k = k'$ survive.

\emph{Step 4, the ball count.} The factors of $P = \prod_j X_{s_j}$ are independent, so $\E[P^2] = \prod_j\E[X_{s_j}^2]$ and $(\E P)^2 = \prod_j(\E X_{s_j})^2$, giving
\begin{equation}
  \frac{\Var(P)}{\bigl(\E P\bigr)^2}
   = \prod_j\frac{\E\bigl[X_{s_j}^2\bigr]}{\bigl(\E X_{s_j}\bigr)^2} - 1
   = \prod_j\Bigl[1+\frac{\Var(X_{s_j})}{\bigl(\E X_{s_j}\bigr)^2}\Bigr] - 1 .
\end{equation}
By Lemma~\ref{lem:F} the numerator is $\Var(X_{s_j}) = (2^m-1)2^{-L}(1-2^{-L})$, and $\E X_{s_j}\ge(2^m-1)2^{-L}$, so each bracket carries a term at most
\begin{equation}
  \frac{(2^m-1)2^{-L}(1-2^{-L})}{\bigl((2^m-1)2^{-L}\bigr)^2}
   = \frac{2^L-1}{2^m-1} \;\approx\; 2^{\,L-m} .
\end{equation}
With $a = n/\log n$ factors and $(1+x)^a-1\approx ax$, we finally obtain
\begin{equation}
  \frac{\Var(P)}{\bigl(\E P\bigr)^2} \;\approx\; a\,2^{\,L-m}
   = \frac{n}{\log n}\cdot n^{\,1-c} = \frac{n^{\,2-c}}{\log n} ,
\end{equation}
whose square root is $n^{(2-c)/2}(\log n)^{-1/2}$.
\end{proof}

\subsection{Fitting the two levels together}\label{sub:l4assemble}

Lemma~\ref{lem:H} splits the bin count into $X_\rho = P/\nu + Y_\rho$ and shows that the ball count $P$ fluctuates, an effect a fixed-count picture leaves out.
This subsection shows that the extra term does not affect the conclusion of Lemma~4, because $P$ takes one value across both branches and cancels from the comparison.

Following Eq.~\eqref{eq:Asplit}, we first sort the $\phi_A$ by the string $\phi_{g_a}$ of bits in $g_a$, and keep the convention that $z^*_h$ is determined by $s_a = \topp(r_{g_a})$.
Fixing $\phi_{g_a}$ contributes a definite $\rho_{g_a}$ to the low part, so this class holds $X_{\rho-\rho_{g_a}}$ configurations with $z^*_l = \rho$.
Substituting Eq.~\eqref{eq:Asplit} into Eq.~\eqref{eq:bridge} then gives
\begin{equation}\label{eq:rhosum}
  \alpha_{z^*_h} = \lambda_4\sum_{\rho\in\Z_\nu}
    \Bigl(\sum_{\phi_{g_a}\,:\ \topp(r_{g_a})=s_a}X_{\rho-\rho_{g_a}}\Bigr)C_\rho .
\end{equation}

Now put $X_{\rho-\rho_{g_a}} = P/\nu + Y_{\rho-\rho_{g_a}}$ into the sum.
The first piece is a constant and depends on neither $\rho$ nor $\phi_{g_a}$, so the two sums separate as
\begin{equation}\label{eq:split}
\begin{split}
  \alpha_{z^*_h} &= 
   \underbrace{\lambda_4\,\frac P\nu\Bigl(\sum_\rho C_\rho\Bigr)X_{s_a}}_{
      \text{main term}} \\
   &+ \underbrace{\lambda_4\!\!\!\sum_{\phi_{g_a}\,:\ \topp(r_{g_a})=s_a}\ \sum_\rho
      Y_{\rho-\rho_{g_a}}C_\rho}_{=:\ R_{h^*}} ,
\end{split}
\end{equation}
where $X_{s_a}$ counts the strings $\phi_{g_a}$ consistent with $s_a$, and $R_{h^*}$ names the error term, one number per branch.

Everything in the main term except $X_{s_a}$ takes the same value on the two branches.
The coefficient $\lambda_4 = (E_0^2+E_1^2)^{-1/2}$ normalises the two branches together and is one number per record.
The ball count $P = \prod_{j<a}X_{W'_j}$ is fixed by the measured $W'$ and by the weights of the groups $j<a$, and neither $h^*$ nor $z^*_h$ occurs in it.
Writing $J := \lambda_4(P/\nu)\sum_\rho C_\rho$, the main term becomes $J\,X_{s_a}$ with the same $J$ on both branches.

For the main term, the quantities $s_a$ and $z^*_h$ differ by the fixed shift $s^* = \sum_{j<a}W'_j + s_a$, so flipping the top bit of $z^*_h$ flips the top bit of $s_a$ and Theorem~\ref{thm:G} applies directly as
\begin{equation}\label{eq:blockdiff}
  \frac{\bigl|X_{s_{a,0}}-X_{s_{a,1}}\bigr|}{\E X_{s_a}}
   \;\le\; 2\,n^{-\left(\frac{c-1}2-1\right)} ,
\end{equation}
where each block departs from the mean by at most $n^{-(\frac{c-1}2-1)}$, with failure probability still $1/n$.
If one of the two blocks is $s_a = 0$, its mean carries the extra empty subset of Eq.~\eqref{eq:blockmean} in a relative amount $2^L/(2^m-1)\approx n^{1-c}$.
For $c\ge12$ that lies far below $n^{-((c-1)/2-1)}$ and is absorbed.
The ball count $P$ is a common factor of the main term and factors out of the difference, so the size of its fluctuation is immaterial.

The remaining quantity to control is $R_{h^*}$, which involves only $Y$.
Lemma~\ref{lem:H} asks that at least one pinned block be nonzero.
All $a-1$ of them vanish with probability $2^{-L(a-1)} = 2^{-(n-\log n)}$, which joins the other failure probabilities and is negligible against them.

The lemma gives $\Var(Y_\rho) = \mu(1-1/\nu)$ with $\mu = \E[P]/\nu$, and $\E[P]$ follows from Eq.~\eqref{eq:blockmean},
\begin{equation}\label{eq:EP}
  \E[P] = \prod_{j<a}\Bigl((2^m-1)2^{-L} + \ind[s_j{=}0]\Bigr)
   \;\approx\; 2^{\,(a-1)(m-L)} .
\end{equation}
With $am = cn$ and $aL = n$, and with $m$ replaced by $m(1-1/(c'\log n))$ to carry the faulty samples of Section~\ref{sub:l4faulty},
\begin{align}
  \mu &= \frac{\E[P]}\nu = 2^{\,cn(1-1/(c'\log n))-2n-(c-2)\log n} , \\
  \sigma &= \sqrt\mu .
\end{align}
Chebyshev bounds the probability that one bin violates $|Y_\rho|\le2^{\,n}\sigma$ by $2^{-2n}$.
Following Section~\ref{sub:prob}, the union bound runs over the $\nu = 2^{\,n-L}$ bins and over all $2^{\,L(a-1)} = 2^{\,n-\log n}$ values of the label vector $W'$.
The exceptional event that survives has probability
\begin{equation}\label{eq:errfail}
  \underbrace{2^{\,n-\log n}}_{W'}\cdot\underbrace{2^{\,n-\log n}}_{\text{bins}}\cdot\, 2^{-2n} \;=\; 2^{-2\log n} \;=\; n^{-2} ,
\end{equation}
which joins the other failure probabilities of order $1/n$.
Outside it every bin sits within a relative $2^{\,n}\sigma/\mu$ of its share.
Taken at a single $\rho$, the summand of Eq.~\eqref{eq:rhosum} is the amplitude $\alpha_{z^*}$ with $z^*_l = \rho$.
The matching summand of $R_{h^*}$ is at most $2^{\,n}\sigma/\mu$ times that amplitude, and Theorem~\ref{thm:E} at $\theta = 2^{\,3n/2}$ bounds it.
Summing over the $\nu$ values of $\rho$,
\begin{equation}\label{eq:Rmax}
  \bigl|R_{h^*}\bigr| \;\le\; R_{\max} \;:=\; \underbrace{\frac{2^{\,n}\sigma}{\mu}}_{\text{one bin}}
    \cdot\underbrace{\nu}_{\text{summands}}
    \cdot\underbrace{2^{\,3n/2}}_{\text{Lemma 3}} .
\end{equation}
Taking logarithms at $c = 12$,
\begin{equation}\label{eq:Rlog}
  \log R_{\max} \;\le\; -\frac{3n}2 \;+\; \frac{6n}{c'\log n} \;+\; 4\log n .
\end{equation}
The middle term is $o(n)$ for every constant faulty rate $c'$, so $R_{\max}\le2^{-3n/2+o(n)}$, comfortably below the $2^{-n}$ the argument needs.

\subsection{Path to the algorithm}
Everything needed for the \emph{difference} of the two aggregated amplitudes is now in place, and the difference is the weaker of the two statements at issue.
Equation~\eqref{eq:split} reads $\alpha_{z^*_h} = J\,X_{s_a} + R_{h^*}$ on each branch, and the prefactor $J$ is the same number for the two members of a pair.
Subtracting the two copies therefore leaves it outside the bracket,
\begin{equation}\label{eq:additive}
  \alpha_{z^*_{h,0}} - \alpha_{z^*_{h,1}}
   = J\bigl(X_{s_{a,0}} - X_{s_{a,1}}\bigr) + \bigl(R_0 - R_1\bigr) .
\end{equation}
The two terms are handled separately.
Equation~\eqref{eq:blockdiff} bounds the count difference by $2n^{-((c-1)/2-1)}\E X_{s_a}$, and each error term obeys $|R_{h^*}|\le R_{\max}$, so
\begin{equation}\label{eq:additivebound}
  \bigl|\alpha_{z^*_{h,0}} - \alpha_{z^*_{h,1}}\bigr|
   \;\le\; 2\,n^{-\left(\frac{c-1}2-1\right)}\,\bigl|J\bigr|\,\E X_{s_a} \;+\; 2R_{\max} .
\end{equation}

Meanwhile, Theorem~\ref{thm:G} in relative form gives $|X_{s_{a,0}}-\E X_{s_a}| \le n^{-((c-1)/2-1)}\E X_{s_a}$ on the same record, leading to
\begin{equation}\label{eq:XvsE}
  \E X_{s_a} \;\le\; \frac{X_{s_{a,0}}}{1-n^{-((c-1)/2-1)}} .
\end{equation}
Multiplying by $|J|$ and reading Eq.~\eqref{eq:split} backwards ($J\,X_{s_{a,0}} = \alpha_{z^*_{h,0}} - R_0$), we have
\begin{equation}\label{eq:JEX}
\begin{split}
  \bigl|J\bigr|\,\E X_{s_a}
   &\;\le\;\frac{\bigl|J\,X_{s_{a,0}}\bigr|}{1-n^{-((c-1)/2-1)}} \\[2pt]
   &\;=\;\frac{\bigl|\alpha_{z^*_{h,0}}-R_0\bigr|}{1-n^{-((c-1)/2-1)}} \\[2pt]
   &\;\le\;\bigl(1+o(1)\bigr)\Bigl(\bigl|\alpha_{z^*_{h,0}}\bigr|+R_{\max}\Bigr),
\end{split}
\end{equation}
the last line by the triangle inequality with $|R_0|\le R_{\max}$, together with $\bigl(1-n^{-((c-1)/2-1)}\bigr)^{-1} = 1+o(1)$ once $c\ge12$.
The corollary to Lemma~3 bounds $|\alpha_{z^*_{h,0}}|$, so every quantity on the right of Eq.~\eqref{eq:additivebound} is already bounded.
Nothing along the way needs $J$ to be large, or even to have a definite sign.

Lemma~4 is stated as a bound on the \emph{ratio} in the form $1\pm O(n^{-((c-1)/2-1)})$, which is strictly stronger than Eq.~\eqref{eq:additivebound}.
To be more specific,
\begin{equation}
  \frac{\alpha_{z^*_{h,0}}}{\alpha_{z^*_{h,1}}}
   = \frac{X_{s_{a,0}} + R_0/J}{X_{s_{a,1}} + R_1/J} ,
\end{equation}
which is close to $1$ only when $|J|\,X_{s_a}$ dominates $R_{\max}$.
A pair whose prefactor satisfies $|J|\,X_{s_a}\lesssim R_{\max}$ has a ratio nowhere near $1$.
Such a pair can occur, since $J$ carries the signed sum $\sum_\rho C_\rho$ over the $\nu = 2^{\,n-L}$ values of $z^*_l$ and no argument keeps such a sum far from zero on every record.

This section therefore establishes Eq.~\eqref{eq:additivebound} and not the ratio.
The rest of the section shows that the difference is all the closing argument ever consumes, so nothing downstream is weakened.
The index $z^*_h$ carries $L = \log n$ bits and therefore takes $n$ values, which form $n/2$ pairs differing only in the top bit $h^*$.
Grouping Eq.~\eqref{eq:diff} by pairs, the quantity to bound is a sum of $n/2$ \emph{differences},
\begin{equation}\label{eq:pairsum}
  \bigl|\hat E_0-\hat E_1\bigr| \;\le\; \sum_{\text{pairs}}
     \bigl|\alpha_{z^*_{h,0}} - \alpha_{z^*_{h,1}}\bigr| ,
\end{equation}
which is the form Eq.~\eqref{eq:additivebound} was stated in.
Both sides carry the normalisation $\lambda_4$, since $\alpha_{z^*_h}$ does.

By \hyperlink{thm:Ep}{Theorem~E$'$} at $\theta = n^{3/2}$, every aggregated amplitude obeys $|\alpha_{z^*_h}|\le n^{3/2}$ except with probability $n^{-3}$, so Eqs.~\eqref{eq:additivebound} and \eqref{eq:JEX} bound each summand by $2n^{-((c-1)/2-1)}(n^{3/2}+R_{\max})(1+o(1)) + 2R_{\max}$.
Since $R_{\max} \le 2^{-3n/2+o(n)}$ is exponentially small, the $R_{\max}$ terms are negligible against everything else and
\begin{equation}\label{eq:E0E1}
\begin{split}
  \bigl|\hat E_0-\hat E_1\bigr| \; &\le\; \frac n2\cdot 2\,n^{3/2} \cdot O\bigl(n^{-((c-1)/2-1)}\bigr)  \\
  &=\; O\bigl(n^{\,7/2-(c-1)/2}\bigr) \;\overset{c\,\ge\,12}{\le}\; O(n^{-2}) ,
\end{split}
\end{equation}
with failure probability $O(1/n)$, and the exponent is consistent with the one obtained from the multiplicative form, so the weaker statement loses nothing.

The two branch amplitudes are now pinned.
Since $\hat E_0^2+\hat E_1^2 = 1$, at least one of $|\hat E_0|$ and $|\hat E_1|$ is at least $1/\sqrt2$.
Equation~\eqref{eq:E0E1} makes the two differ by $O(n^{-2})$, so both are within $O(n^{-2})$ of $1/\sqrt2$ and carry the same sign.
The normalisation cancels from the ratio, so
\begin{equation}
  q \;=\; \frac{\bigl|E_0-E_1\bigr|}{\bigl|E_0+E_1\bigr|}
    \;=\; \frac{\bigl|\hat E_0-\hat E_1\bigr|}{\bigl|\hat E_0+\hat E_1\bigr|}
    \;=\; O\bigl(n^{-2}\bigr) ,
\end{equation}
and Eq.~\eqref{eq:p} gives $p = 1-O(n^{-4})$, which exceeds $1/2$ by a constant.
Polynomially many repetitions and a majority vote then return $d_n$, and the recursion on the remaining bits of $d$ finishes the algorithm.

\section{Numerical checks}\label{sec:num}

\begin{figure*}[t]
\centering
\includegraphics[width=\linewidth]{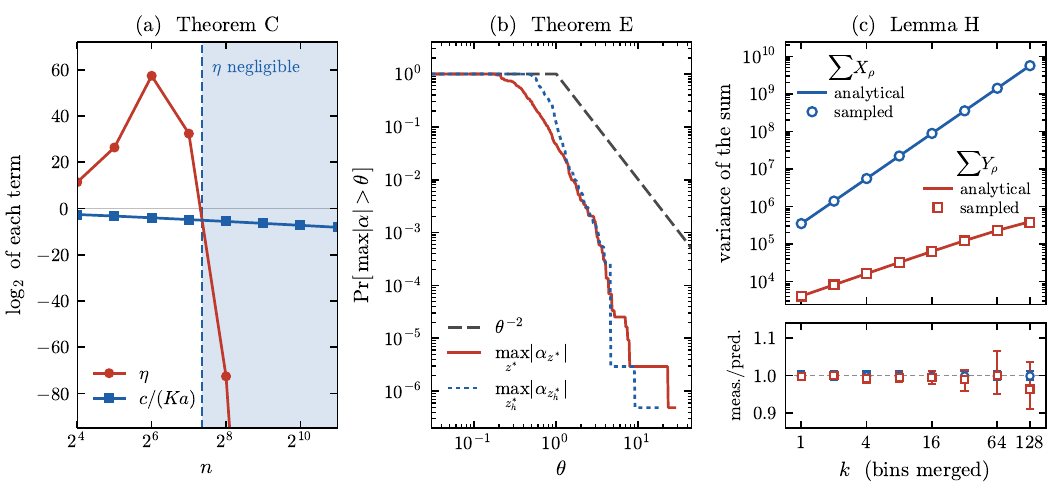}
\caption{One numerical check for each of the three results.
(a) The two terms that bound $\Var[\mathcal Z]/\E[\mathcal Z]^2$ in
Eq.~\eqref{eq:varZ}, at $c = 12$ and $K = 1.5c$.
The correlation term $\eta$ is computed exactly and drops below the Poisson
term at $n\approx164$.
(b) Measured tail of the amplitude against the bound of Theorem~\ref{thm:E},
taken over the full cube of $2^{10}$ measurement outcomes of a toy instance
with $n = 6$ and $Q = 20$.
(c) Variance of $k$ merged bins in the second balls-in-bins model, at $n = 12$,
$m = 10$, $L = 3$, $\nu = 512$ and three pinned groups, from $2\times10^5$
samples.
The upper curve carries the ball-count term and the lower one does not.
The lines are computed from Lemma~\ref{lem:H}, with $\mu$ and $\Var(P)$ taken
from Eq.~\eqref{eq:blockmean} and Eq.~\eqref{eq:varP}, while the points are
measured from the samples.
The strip below gives the measured value over the computed one, with $95\%$
bootstrap intervals.}
\label{fig:num}
\end{figure*}

The parameters of \Sim\ put a direct simulation of the algorithm out of reach, since $c\ge12$ makes $Q = Kn^{c+1}$ exceed $10^{16}$ already at $n = 16$.
The three results proved above are statements about classical counting quantities, and those can be evaluated in full at small parameters.
Figure~\ref{fig:num} collects one check for each.

Panel~(a) plots the two terms of Eq.~\eqref{eq:varZ} against $n$ at the parameters of \Sim.
The correlation coefficient $\eta$ is obtained in closed form from Eq.~\eqref{eq:ET} in exact rational arithmetic, using the averages of Eqs.~\eqref{eq:fejer}, \eqref{eq:extraavg} and \eqref{eq:extraavg2}.
The value agrees with a direct sum over $k$ for $n\le16$.
The curve rises while $2^m$ exceeds $N$ and falls once $2^m$ drops below it, crossing the Poisson term $c/(Ka)$ at $n\approx164$.
Beyond that point $\eta$ contributes nothing and the bound of Theorem~\ref{thm:C} is carried by $1/\E[\mathcal Z]$ alone, which leaves $\Pr[\mathcal Z\ge a]$ at least $0.81$ for $n = 256$ and at least $0.96$ for $n = 2048$.

Panel~(b) takes a toy instance small enough to enumerate, with $n = 6$ and $Q = 20$ sample bits cut into $G = 4$ groups of $m = 5$, of which $a = 2$ sit on the $A$ side.
Every one of the $2^{10}$ strings $\phi'_B$ is summed over, which imposes Hypothesis~(U) by construction.
The measured tail runs a factor of $20$ below $\theta^{-2}$ at $\theta = 1$ and a factor of $10^3$ below it at $\theta = 16$, so the bound of Theorem~\ref{thm:E} is comfortable as well as valid.

Panel~(c) tests Lemma~\ref{lem:H} by merging $k$ adjacent bins.
Since $Y_\rho$ and $P$ are uncorrelated, Eq.~\eqref{eq:covY} predicts $\Var(\sum_{\rho<k}Y_\rho) = \mu k(1-k/\nu)$, while the same sum of the bin counts themselves carries the further term $k^2\Var(P)/\nu^2$.
Both are computed from the lemma alone, with $\mu$ taken from Eq.~\eqref{eq:blockmean} and $\Var(P)$ from Eq.~\eqref{eq:varP}.
The measured values come from $2\times10^5$ sampled runs and agree with both.
The gap between the two curves is the term a fixed ball count leaves out, and it exceeds the scattering term by orders of magnitude.
Therefore, the ball count governs the bin variance completely, and Section~\ref{sub:l4assemble} is where it cancels between the two branches.

None of these checks bears on Hypothesis~(U), since panel~(b) imposes it by construction and the other two never meet it.
The agreement reported here therefore says nothing about the one gap recorded in Section~\ref{sec:rests}.

\section{Conclusions and Discussions}\label{sec:rests}

In this paper we have restated the three sketched lemmas of \Sim in a form that admits a single reading and have proved each of them, as summarized in Table~\ref{tab:final}.
Only Lemma~4 comes out weaker than the original claim, with a bound on the difference of the two branch amplitudes in place of their ratio, and Section~\ref{sub:l4assemble} shows that the closing argument consumes only the difference.

\begin{table*}[t]

\renewcommand{\arraystretch}{1.15}
\centering
\caption{The four lemmas in the forms proved here.}
\label{tab:final}
\begin{tabularx}{\linewidth}{@{}lL@{}}
\toprule
In \Sim & Status \\
\midrule
Lemma~1 & Holds, with probability tending to one
  (Theorem~\ref{thm:C}), stronger than the constant claimed \\
Lemma~2 & Exact as published, with a two-line proof
  (Section~\ref{sec:setup}) \\
Lemma~3 and Corollary & The amplitude bound holds at every threshold
  (Theorem~\ref{thm:E}, \hyperlink{thm:Ep}{E$'$}) and needs no well-behavedness
  hypothesis, so the first half of the lemma never enters and is recorded only
  in Appendix~\ref{app:E} \\
Lemma~4 & Both counting estimates hold, the first with the constant \Sim\
  displays (Theorem~\ref{thm:G}), the second together with the term a fixed ball
  count leaves out (Lemma~\ref{lem:H}), and the assumption that $g_a$ carries no
  faulty samples can be dropped (Section~\ref{sub:l4faulty}).
  The comparison of the two branches is proved in the weaker form of
  Eq.~\eqref{eq:additivebound}, which is the form the
  closing argument uses.\\
\bottomrule
\end{tabularx}
\end{table*}

One condition survives all of this.
Hypothesis~(U) asks that the $A/B$ partition be fixed independently of $\phi'$, so that $\phi'_B$ ranges over the whole cube $\{0,1\}^{Q_B}$ with the rest of the record held fixed.
However, the algorithm itself does not supply it.
It requires $A$ to consist of the first $a$ all-zero groups, which makes the partition a function of $\phi'$.
The hypothesis enters at a single step in Lemma~\ref{lem:D},
\begin{equation}\label{eq:parsevalU}
  \sum_{\phi'_B\in\{0,1\}^{Q_B}}E_{h^*}(\phi')^2
   = 2^{Q_B}\sum_{\phi_B}\mathcal N_{h^*}(\phi_B)^2 ,
\end{equation}
where the character sum over the whole cube collapses to the diagonal $\phi_B = \psi_B$ and leaves the factor $2^{Q_B}$.

Narrowing that range is not a perturbation of the identity.
A group that the rule forces away from the all-zero block contributes
\begin{equation}\label{eq:punctured}
  \sum_{\phi'_j\ne0}(-1)^{(\phi_j\oplus\psi_j)\cdot\phi'_j}
   \;=\; 2^m\delta_{\phi_j\psi_j} - 1
\end{equation}
in place of $2^m\delta_{\phi_j\psi_j}$, so Eq.~\eqref{eq:parsevalU} acquires one term for every subset of the forced groups and the off-diagonal ones no longer cancel.
Controlling them would need the autocorrelation of $\mathcal N_{h^*}$, which is the one object Lemma~\ref{lem:D} is built to avoid.
Nor does it help that the discarded strings are few, because Eq.~\eqref{eq:parsevalU} holds only after the whole cube has been summed.
Its left side collects squares of signed sums and its right side squares of counts, and the two are unrelated at a fixed $\phi'_B$.

The two demands on $A$ pull against each other, and that is why a change of rule does not remove the point.
The groups of $A$ have to be all-zero, since only then do they carry no sign and leave the $A$-side factor an unsigned count that both branches of the final interference share.
Selecting them therefore means looking at $\phi'$, and that look costs the freedom Hypothesis~(U) asks for.
Fixing the index set before Step~4 would grant the hypothesis and give up the all-zero property.
Therefore, establishing these four lemmas does not by itself establish the correctness of the algorithm.

\appendix
\section{The second moment of a group's subset-sum count}\label{app:B}

This appendix derives Eq.~\eqref{eq:uu} and Eq.~\eqref{eq:ET}, together with the two averages of the Fej\'er kernel quoted in Eq.~\eqref{eq:fejer}.

\subsection{The four classes of pairs}

Write $\Sigma(\phi) = \sum_{i\in g_j}\phi_iy_i \bmod N$, so that $n_j[r] = \sum_{\phi\in\{0,1\}^m}\ind[\Sigma(\phi)=r]$.
Multiplying two of these gives a double sum over pairs of subsets,
\begin{equation}
  \E\bigl[n_j[r]\,n_j[r']\bigr]
   = \sum_{\phi,\psi}\Pr_y\bigl[\Sigma(\phi)=r,\ \Sigma(\psi)=r'\bigr] .
\end{equation}
For a single subset, $\Pr_y[\Sigma(\phi)=r]$ depends only on whether $\phi$ is empty.
The empty subset gives $\Sigma\equiv0$ with certainty, and any nonempty subset gives $\Sigma$ uniform on $\Z_N$, since a sum of independent uniform values is uniform.
Hence
\begin{align}
  \E\bigl[n_j[r]\bigr] = \underbrace{\ind[r=0]}_{\phi=0} + \underbrace{\beta}_{\phi\ne0}, \qquad \beta = \frac{2^m-1}{N} .
\end{align}
For a pair, the same distinction applies to each member, and the four resulting classes are collected in Table~\ref{tab:classes}.
In the last case, $\phi$ and $\psi$ are distinct and neither is empty, so the two sums are jointly uniform on $\Z_N^2$.
Adding the four contributions gives Eq.~\eqref{eq:uu}.

\begin{table*}[t]

\renewcommand{\arraystretch}{1.15}
\centering
\caption{The four classes of subset pairs and what each contributes to
Eq.~\eqref{eq:uu}.}
\label{tab:classes}
\begin{tabularx}{\linewidth}{@{}lllLL@{}}
\toprule
Class & Condition & Pairs & $\Pr[\Sigma_\phi{=}r,\Sigma_\psi{=}r']$
  & Contribution \\
\midrule
atom & $\phi=\psi=0$ & $1$ & $\ind[r{=}0]\ind[r'{=}0]$
  & $\ind[r{=}r'{=}0]$ \\
cross-atom & exactly one empty & $2(2^m-1)$
  & $\tfrac1N\ind[r{=}0]$ or $\tfrac1N\ind[r'{=}0]$
  & $\beta\bigl(\ind[r{=}0]+\ind[r'{=}0]\bigr)$ \\
diagonal & $\phi=\psi\ne0$ & $2^m-1$ & $\tfrac1N\delta_{rr'}$
  & $\beta\,\delta_{rr'}$ \\
constant & $\phi\ne\psi$, neither empty & $(2^m-1)(2^m-2)$
  & $\dfrac1{N^2}$ & $\dfrac{4^m-3\cdot2^m+2}{N^2}$ \\
\bottomrule
\end{tabularx}
\end{table*}

\subsection{Carrying the four classes through \texorpdfstring{$T(k)$}{T(k)}}

Substitute Table~\ref{tab:classes} into $T_j(k) = \sum_s\sum_{r,r'\in s}n_j[r]n_j[r']\,\omega^{-k(r-r')}$, where the inner sum runs over the $r$ with the top $L$ bits fixed and the low $n-L$ bits free.

The diagonal term contributes $\beta\sum_s\sum_{r\in s}1 = \beta N = 2^m-1$, and the atom contributes $1$, the two together making $2^m$.

For the constant term, $\Z_N$ is cut into $2^L$ blocks of length $\nu$, and the bottom one is $B_0 = \{0,1,\dots,\nu-1\}$.
Its indicator has transform
\begin{equation}
  \hat\ind_{B_0}(k) = \sum_{r=0}^{\nu-1}\omega^{-kr}
   = \begin{cases}\nu, & k=0,\\[3pt]
       \dfrac{1-\omega^{-k\nu}}{1-\omega^{-k}}, & k\ne0,\end{cases}
\end{equation}
and taking the squared modulus with $|1-e^{-i\theta}|^2 = 4\sin^2(\theta/2)$ and $\nu/N = 2^{-L}$
\begin{equation}
  D_\nu(k) := \bigl|\hat\ind_{B_0}(k)\bigr|^2
   = \frac{\sin^2(\pi k\nu/N)}{\sin^2(\pi k/N)}
   = \frac{\sin^2(\pi k/2^L)}{\sin^2(\pi k/N)} ,
\end{equation}
the unnormalised Fej\'er kernel of order $\nu$.
Every block is a translate of $B_0$ and translation contributes only a phase, so $\sum_s|\sum_{r\in s}\omega^{-kr}|^2 = 2^LD_\nu(k)$, and the constant term contributes $\beta_2D_\nu(k)$.

For the cross-atom terms, one of the two indices is pinned to $0$ and the other runs over $B_0$.
The two cases give complex conjugate expressions, so together they contribute $2\beta\,\mathrm{Re}\,\hat\ind_{B_0}(k) = \beta_1\,\mathrm{Re}\,\hat\ind_{B_0}(k)$.
Collecting the four gives Eq.~\eqref{eq:ET}.

\subsection{Averages over the frequency line}

Parseval with $f = g = \ind_{B_0}$ gives the first
\begin{equation}
  \bigl\langle D_\nu\bigr\rangle_k
   = \frac1N\sum_k\bigl|\hat\ind_{B_0}(k)\bigr|^2
   = \sum_r\ind_{B_0}[r]^2 = \nu .
\end{equation}
For the second, we use $\widehat{f\star f}(k) = |\hat f(k)|^2$ from the proof of Lemma~\ref{lem:B} to give
\begin{equation}
  \frac1N\sum_k D_\nu(k)\,\omega^{k\delta}
   = \bigl(\ind_{B_0}\star\ind_{B_0}\bigr)[\delta] =: \Lambda[\delta] ,
\end{equation}
and since $\ind_{B_0}$ takes values in $\{0,1\}$,
\begin{equation}
\begin{split}
  \Lambda[\delta] &= \#\bigl\{r : r\in B_0 \text{ and } r-\delta\in B_0\bigr\} \\
  &= \bigl|B_0\cap(B_0+\delta)\bigr| = (\nu-|\delta|)_+ ,
\end{split}
\end{equation}
the overlap of two intervals of length $\nu$ offset by $\delta$, with $(x)_+ = \max(x,0)$.
Applying Parseval once more to $\Lambda$ gives
\begin{equation}
\begin{split}
  \bigl\langle D_\nu^2\bigr\rangle_k &= \frac1N\sum_k\bigl|\hat\Lambda(k)\bigr|^2 = \sum_\delta \Lambda[\delta]^2 \\
  &= \underbrace{\nu^2}_{\delta=0} + \underbrace{2\sum_{j=1}^{\nu-1}j^2}_{\delta=\pm1,\dots,\pm(\nu-1)} = \frac{\nu\bigl(2\nu^2+1\bigr)}{3} \;\le\; \nu^3 ,
\end{split}
\end{equation}
which is Eq.~\eqref{eq:fejer}.

Three further averages are used by the numerical check of Section~\ref{sec:num}, and the same two tools deliver them.
Writing $f = \ind_{B_0}$ and reading the inverse transform at $\delta = 0$ gives $\langle\hat f\rangle_k = f[0] = 1$, so $\langle\mathrm{Re}\,\hat f\rangle_k = 1$.
Since $\hat f(k)^2 = \widehat{f\star f}(k)$ up to the reflection $r\mapsto-r$, the same reading gives $\langle\hat f^2\rangle_k = \#\{r\in B_0 : -r\in B_0\} = 1$, where the only such $r$ is $0$ because $\nu\le N/2$.
Expanding $(\mathrm{Re}\,\hat f)^2 = \tfrac14(\hat f^2 + 2|\hat f|^2 + \overline{\hat f}^2)$ and using $\langle|\hat f|^2\rangle_k = \nu$, we have
\begin{equation}\label{eq:extraavg}
  \bigl\langle(\mathrm{Re}\,\hat\ind_{B_0})^2\bigr\rangle_k = \frac{1+2\nu+1}4 = \frac{\nu+1}2 .
\end{equation}
For the last one, $\hat f^2\overline{\hat f}$ at $\delta = 0$ counts the triples of $B_0$ with $r_1+r_2 = r_3$, so
\begin{equation}\label{eq:extraavg2}
\begin{split}
  \bigl\langle D_\nu\,\mathrm{Re}\,\hat\ind_{B_0}\bigr\rangle_k
   &= \#\bigl\{(r_1,r_2)\in B_0^2 : r_1+r_2\le\nu-1\bigr\} \\
   &= \frac{\nu(\nu+1)}2 .
\end{split}
\end{equation}

\section{The first half of Lemma 3}\label{app:E}

This appendix records the other half of the statement quoted at the start of Section~\ref{sec:l3}.
Its role there is to make the second half non-vacuous, and the proof given in Section~\ref{sub:l3F} does not need that service.

The signed count of Section~\ref{sec:setup} is $E_{h^*} = t^+_M-t^-_M$, so being well-behaved for $h^*$ says $\bigl|E_{h^*}\bigr| \ge \varepsilon\sqrt{t_{M,h^*}}$.
The dependence of $t_M$ on $h^*$ is written explicitly, since $T_M$ is by definition the set consistent with the record and a given $h^*$, and the two branches carry their own counts.
Failing to be well-behaved is the same as carrying little weight.
The predicate of Definition~1 contains $\Omega(\cdot)$, which names a class of functions and not a number, so the clause becomes a definition only once a constant is fixed.
That constant matters, because the failure probability scales as its square.
The threshold is therefore written as $\varepsilon\sqrt{t_{M,h^*}}$ and the conclusion as $\varepsilon^2$.

One consequence of Lemma~\ref{lem:D} is needed and was not drawn there.
The function $\mathcal N_{h^*}(\phi_B)$ takes nonnegative integer values, so $\mathcal N_{h^*}^2 \ge \mathcal N_{h^*}$ pointwise, and summing with $\sum_{\phi_B}\mathcal N_{h^*}(\phi_B) = t_{M,h^*}$ turns Eq.~\eqref{eq:parseval} into
\begin{equation}\label{eq:DtM}
  \sum_{\phi'_B\in\{0,1\}^{Q_B}}E_{h^*}(\phi')^2 \;\ge\; 2^{Q_B}\,t_{M,h^*} .
\end{equation}

\begin{theorem}\label{thm:I}
Assume Hypothesis~(U).
Fix $\varepsilon>0$ and call the superposition well-behaved for $h^*$ when $|E_{h^*}(\phi')| \ge \varepsilon\sqrt{t_{M,h^*}}$.
The probability that it is well-behaved for neither value of $h^*$ is then at most $\varepsilon^2$.
Taking $\varepsilon = 2^{-n/2}$ gives $2^{-n}$, which is the conclusion of the first half of Lemma~3 of \Sim.
\end{theorem}

\begin{proof}
Fix everything in the record except $\phi'$, i.e., $Y$, $z'$, $W'$, $S$, $h'$ and the $A/B$ partition, and let $\phi'_B$ vary.
The weight of a record is proportional to $E_0(\phi')^2+E_1(\phi')^2$, so the quantity to bound is
\begin{align}
\begin{aligned}
  &\;\Pr\bigl[\text{neither } h^* \text{ well-behaved}\bigr]\\
   =&\; \frac{\sum_{\phi'_B\ \text{failing both}}
       \bigl(E_0(\phi')^2+E_1(\phi')^2\bigr)}
      {\sum_{\phi'_B\in\{0,1\}^{Q_B}}\bigl(E_0(\phi')^2+E_1(\phi')^2\bigr)} .
\end{aligned}
\end{align}
A record failing both satisfies $E_0(\phi')^2 < \varepsilon^2t_{M,0}$ and $E_1(\phi')^2 < \varepsilon^2t_{M,1}$, so $E_0(\phi')^2+E_1(\phi')^2 < \varepsilon^2(t_{M,0}+t_{M,1})$.
There are at most $2^{Q_B}$ such strings, so the numerator is below $2^{Q_B}\varepsilon^2(t_{M,0}+t_{M,1})$.
Applying Eq.~\eqref{eq:DtM} to each branch, the denominator is
\begin{equation}
  2^{Q_B}\sum_{\phi_B}\bigl(\mathcal N_0(\phi_B)^2+\mathcal N_1(\phi_B)^2\bigr)
   \;\ge\; 2^{Q_B}\bigl(t_{M,0}+t_{M,1}\bigr) .
\end{equation}
Dividing the numerator by the denominator, both $2^{Q_B}$ and $t_{M,0}+t_{M,1}$ cancel and $\varepsilon^2$ remains.
\end{proof}

\emph{Note added}. Since \Sim appeared, its lemmas have been the subject of active public scrutiny.
We became aware of parts of that discussion while this work was in preparation.
Where our findings overlap with analyses circulated elsewhere we claim no priority.
Our results were obtained independently and are stated so that each can be checked against \Sim directly.

\begin{acknowledgments}
This work is supported by the National Natural Science Foundation of China (NSFC) (Grant No.~12475022 and No.~125B2100) and the Quantum Science and Technology - National Science and Technology Major Project (Grant No.~2021ZD0302100).
\end{acknowledgments}


\end{document}